\documentclass[a4paper, 11pt]{article}
\pdfoutput=1

\usepackage[utf8]{inputenc}

\usepackage[margin=1.16in]{geometry}

\usepackage{amsmath,amsthm,amssymb}
\usepackage{mathtools}

\usepackage[shortlabels]{enumitem}

\usepackage[dvipsnames]{xcolor}
\usepackage[breaklinks,colorlinks=true,linkcolor=Blue, citecolor=OliveGreen, urlcolor=RoyalBlue]{hyperref}

\DeclareUnicodeCharacter{1EF3}{\`y}
\usepackage[backend=biber,isbn=false,url=true,doi=true,maxcitenames=3,maxbibnames=99]{biblatex}
\bibliography{strings-long}
\bibliography{bibliography}

\usepackage[capitalise]{cleveref}
\usepackage{xspace}

\usepackage{multirow}

\usepackage[colorinlistoftodos,prependcaption,color=green!30,textsize=footnotesize]{todonotes}

\theoremstyle{plain}
\newtheorem{theorem}{Theorem}[section]
\newtheorem{conjecture}{Conjecture}
\newtheorem{corollary}[theorem]{Corollary}

\newtheorem{lemma}[theorem]{Lemma}

\newtheorem{remark}[theorem]{Remark}

\usepackage{algorithm}
\usepackage[noend]{algpseudocode}

\theoremstyle{definition}

\newtheorem*{claim*}{Claim}

\newcommand{\A}{\mathcal A}
\newcommand{\G}{\mathcal G}
\newcommand{\Hh}{\mathcal H}
\newcommand{\F}{\mathcal F}

\newcommand{\Pp}{\mathcal P}
\newcommand{\bigO}{O}
\newcommand{\eps}{\varepsilon}
\newcommand{\prob}{\mathbb{P}}
\newcommand{\expect}{\mathbb{E}}
\newcommand{\variance}{\mathbb{D}}

\newcommand{\w}{w}
\newcommand{\ie}{i.e.,\xspace}
\usepackage{graphicx,scalerel}
\newcommand{\capped}[1]{\hstretch{2}{\hat{\hstretch{.5}{#1}}}}
\newcommand{\aas}{a.a.s.\@\xspace}

\newcommand{\lnnf}{\frac{\log n}{n}}

\DeclarePairedDelimiter{\abs}{\lvert}{\rvert}
\DeclareMathOperator*{\argmin}{arg\,min}
\DeclareMathOperator{\dd}{d\!} 
\DeclarePairedDelimiter\ceil{\lceil}{\rceil}
\DeclarePairedDelimiter\floor{\lfloor}{\rfloor}

\newcommand{\ANY}{\texttt{ANY}}
\newcommand{\CO}{\texttt{CO}}

\title{Sharp Thresholds in Random Simple Temporal Graphs\footnote{An extended abstract of this work appeared in the Proceedings of the 62nd Annual Symposium on Foundations of Computer Science (FOCS) \cite{CRRZ22}. The work of the first and the second authors was supported by the French ANR, project ANR-22-CE48-0001 (TEMPOGRAL). The work of the second author was also supported by the European Research Council (ERC) under the Horizon 2020 research and innovation program (grant agreement No 787367 (PaVeS)). The work of the third author was supported by the German Research Foundation (DFG), project MATE (NI 369/17).}}
\author{
Arnaud Casteigts\thanks{University of Geneva, Switzerland. Email: \texttt{arnaud.casteigts@unige.ch}. ORCID: 0000-0002-7819-7013. Part of this project was done while the author was at LaBRI, CNRS, University of Bordeaux, France.}
\and 
Michael Raskin\thanks{LaBRI, CNRS, University of Bordeaux, France. Email: \texttt{mraskin@u-bordeaux.fr}. ORCID:0000-0002-6660-5673. Part of this project was done while the author was at Department of Informatics, Technical University of Munich, Germany.}
\and
Malte Renken\thanks{Technical University of Berlin, Germany. Email: \texttt{m.renken@tu-berlin.de}. ORCID: 0000-0002-1450-1901.}
\and
Viktor Zamaraev\thanks{Department of Computer Science, University of Liverpool, United Kingdom. Email: \texttt{Viktor.Zamaraev@liverpool.ac.uk}. ORCID: 0000-0001-5755-4141.}
}
\date{}

\begin{document}

\maketitle

\begin{abstract}
  A graph whose edges only appear at certain points in time is called
  a temporal graph (among other names). Such a graph is temporally
  connected if each ordered pair of vertices is connected by a path
  which traverses edges in chronological order (\ie a temporal path).
  In this paper, we consider a simple model of random temporal graph,
  obtained from an Erd{\H{o}}s–R{\'e}nyi random graph~$G\sim G_{n,p}$
  by considering a random permutation $\pi$ of the edges and
  interpreting the ranks in $\pi$ as presence times.

  We give a thorough study of the temporal connectivity of such graphs
  and derive implications for the existence of several kinds of sparse
  spanners. It turns out that temporal reachability in this model
  exhibits a surprisingly regular sequence of thresholds. In
  particular, we show that, at $p=\log n/n$, any fixed pair of
  vertices can \aas reach each other; at $2\log n/n$, at least one
  vertex (and in fact, any fixed vertex) can \aas reach all others;
  and at $3\log n/n$, all the vertices can \aas reach each other, \ie
  the graph is temporally connected. Furthermore, the graph admits a
  temporal spanner of size $2n + o(n)$ as soon as it becomes
  temporally connected, which is nearly optimal as $2n-4$ is a lower
  bound. This result is quite significant because temporal graphs do
  not admit spanners of size $\bigO(n)$ in general (Kempe, Kleinberg,
  Kumar, STOC 2000). In fact, they do not even always admit spanners
  of size $o(n^2)$ (Axiotis, Fotakis, ICALP 2016). Thus, our result
  implies that the obstructions found in these works, and more
  generally, any non-negligible obstruction, is statistically
  insignificant: \emph{nearly optimal spanners always exist in random
  temporal graphs}.

  All the above thresholds are sharp. Carrying the study of temporal
  spanners a step further, we show that pivotal spanners---\ie
  spanners of size $2n-2$ composed of two spanning trees glued at a single
  vertex (one descending in time, the other ascending subsequently)---exist \aas at $4\log n/n$, this threshold being also
  sharp. Finally, we show that optimal spanners (of size $2n-4$) also
  exist \aas at $p=4 \log n/n$. Whether this value is a
  sharp threshold is open; we conjecture that it is.

  For completeness, we compare the above results to existing results
  in related areas, including edge-ordered graphs, gossip
  theory, and population protocols, showing that our results can be
  interpreted in these settings as well, and that in some cases, they
  improve known results therein. Finally, we discuss an intriguing
  connection between our results and Janson's celebrated results on
  percolation in weighted graphs.
\end{abstract}
\medskip



\section{Introduction}

A temporal graph is a graph whose edges are present only at certain times. These graphs can be modeled in various ways, a classical option being as an edge-labeled graph $\G = (G, \lambda)$, where~$G$ is a standard graph and~$\lambda$ encodes the presence times of the edges.
Temporal graphs have been extensively studied in the past two decades, motivated by the modeling of dynamic networks in areas like social network analysis, communication networks, epidemics, transportation, and biology.
As theoretical objects, these graphs pose a number of fundamental questions. In particular, the redefinition of classical graph concepts in a temporal setting often leads to subtle, yet radical, differences. 
A canonical example is that of a \emph{temporal path}, which is a path running over a non-decreasing sequence of timestamps (or increasing, in the case of strict temporal paths).
Clearly, such paths are not symmetric, even with undirected edges, and they are also not transitive---the fact that vertex~$u$ can reach vertex~$v$ and $v$~can reach~$w$ does not imply that $u$~can reach~$w$.
Non-transitivity makes temporal paths quite different even from classical paths in directed graphs, with a strong impact on problems and algorithms.
Among the early examples, it was shown in~\cite{KKK02} that deciding whether a set of $k$ vertex-disjoint temporal paths exist between two given vertices is NP-complete, whereas the analog problem in static graphs is polynomial-time solvable.
Likewise, computing a maximum connected component based on temporal paths is essentially as hard as finding a maximum clique in static graphs~\cite{BF03}. Since then, many tractable problems have been shown to have intractable analogs in temporal graphs (see, e.g., \cite{akrida2020temporal, waiting-time, minimization, separators, diameter, mertzios2020computing, mertzios2021sliding}).

In this landscape, the status of spanning trees and other spanning structures remains unsettled.
In static graphs, a spanning tree is a connected spanning subgraph without cycle, and therefore consists of $n-1$~edges.
It always exists when the graph is connected, and computing one is straightforward.
In contrast, \textcite{KKK02} observe that the size of a minimum temporal spanner varies among (temporally connected) temporal graphs, and in some cases, only spanners of superlinear size exist.
For example, hypercubes can be time-labeled in such a way that none of their $\Theta(n \log n)$~edges can be removed without breaking temporal connectivity (\ie they are minimally connected).
More surprisingly, \textcite{AF16} construct an infinite family of minimally connected temporal graphs having $\Theta(n^2)$~edges. 
In other words, even \emph{sparse} analogs of spanning trees do not exist unconditionally in temporal graphs. (Incidentally, finding a minimum spanner is computationally difficult, namely APX-hard~\cite{Akr+17}.)

On the positive side, \textcite{CPS19} showed that temporal graphs whose underlying graph is complete always admit spanners of size~$\bigO(n \log n)$. 
Whether a spanner of size~$\bigO(n)$ always exists in complete temporal graph is still open.
However, a general result in gossip theory~\cite{Bumby1979} implies that $2n-4$~time~labels are required to achieve temporal connectivity in any underlying graph (provided only that adjacent labels are distinct).
Thus, spanners of size~$2n-4$ are optimal, and we call a spanner of size~$2n + o(n)$ \emph{nearly optimal}.

All these results were obtained for temporal graphs whose edges have exactly one presence time, \ie \emph{simple} temporal graphs. Due to their simplicity, these graphs prove to be good prototypes for studying connectivity problems. Moreover, the above results extend to general temporal graphs, either by containment (for negative results), or by the fact that the spanning property is preserved under the addition of time labels (for positive results). In the following, we restrict our attention to simple temporal graphs and refer the reader to~\cite{casteigts2012time} and~\cite{holme2019temporal} for background on general models.

\subsection{Contributions}

The present work investigates temporal reachability and temporal spanners from a \emph{probabilistic} point of view. To this end,
we study a simple and natural model of random temporal graphs which can be seen as a temporal analog of Erd{\H{o}}s–R{\'e}nyi random graphs.
More precisely, given a number of nodes~$n$ and a probability~$p$, a \emph{random simple temporal graph} (RSTG, for short) is obtained by taking an Erdős-Rényi random graph $(V, E) = G(n,p)$ and assigning to each edge~$e \in E$ a \emph{unique} presence time $\lambda(e) \in \{1, \ldots, \abs{E}\}$ according to a random permutation of all edges.
The reachability in this model only depends on the relative order of edge labels;
it can thus be equivalently defined by having $\lambda:E \to [0,1]$ assign to each edge a presence time chosen uniformly at random (and independently) from the unit interval (with probability~$1$, all edges receive distinct labels).
As will be discussed below, RSTGs are closely related to a number of other models.
For example, they are mathematically equivalent to random edge-ordered graphs, although the interpretation of labels as time motivates the study of specific questions.
RSTGs are also related to classical processes in gossip theory and population protocols.
In a certain sense, they \emph{encode} information pertaining to the ordering of interactions in these models.
A significant difference, that also distinguishes RSTGs from stochastic models of dynamic networks like edge-markovian evolving graphs~\cite{BCF11,clementi08}, is that there are no repeated interactions among the same nodes.
This creates dependencies between past and future events that make RSTGs typically more difficult to handle.

\subsection*{Overview of the results}

Inspired by classical results on Erd{\H{o}}s–R{\'e}nyi graphs, we investigate natural analogs of connectivity thresholds in RSTGs.
The fact that temporal reachability is neither transitive nor symmetric implies a number of different thresholds pertaining to gradual forms of connectivity. We consider the following properties: (1) a fixed node can reach another fixed node asymptotically almost surely (\aas); (2) \aas at least one node can reach all others; (3) a fixed node can \aas reach all the others; and (4) every node can reach all the others \aas (\ie the graph is \emph{temporally connected} \aas).

The first set of results in this paper is a complete characterization of these thresholds, which follow a strikingly regular set of incremental values as follows.
The first property occurs at $p=\log n / n$ (where $\log$ is the natural logarithm), the second and third occur at $p=2 \log n / n$, and the fourth at $p=3 \log n / n$.
All these thresholds are sharp; that is, for the threshold factors $c\in\{1,2,3\}$, the properties \aas do not hold at $p=(c-o(1)) \log n / n$ and hold \aas at $p=(c+o(1)) \log n / n$. 
These thresholds are summarized in \cref{tab:thresholds}.

\begin{table}
\centering
\begin{tabular}{|l|c|c|c|}
	\hline
	\textbf{Property} & \textbf{Shorthand} & \textbf{Sharp threshold} & \textbf{Reference} \\
	\hline
	Point-to-point Reachability & $\forall u\, \forall v \:\textbf{a.a.s.}\: u\leadsto v$
            & $\log n/n$ & \cref{th:pathUV} \\
    \hline
        First Temporal Source & {$\textbf{a.a.s.}\: \exists u\, \forall v\: u\leadsto v$}
    	& $2\log n / n$ & \cref{th:firstSource}  \\
	\hline
	Temporal Source &$\forall u \:\textbf{a.a.s.}\: \forall v\: u\leadsto v$
            & $2 \log n/n$ & \cref{th:avgSource} \\
	\hline
	Temporal Connectivity & $\textbf{a.a.s.}\: \forall u\, \forall v\: u\leadsto v$
            & $3\log n / n$ & \cref{th:tempConnectivity} \\
	\hline
\end{tabular}
\caption{Sharp thresholds for connectivity properties in random temporal graphs. 
The notation $u \leadsto v$ denotes the existence of a temporal path between $u$ and $v$.}
\label{tab:thresholds}
\end{table}

Subsequently, we consider the densities at which various types of spanners start to exist.
In particular, we are interested in the emergence of (5) nearly optimal spanners (of size~$2n+o(1)$);
(6) pivotal spanners (of size~$2n-2$), composed of two spanning trees glued at a single vertex, one descending in time, the other ascending subsequently;
and (7) optimal spanners (of size~$2n-4$).
Clearly, none of these structures can exist before $p=3 \log n /n$, because the corresponding graphs are not even temporally connected.
Interestingly, nearly optimal spanners emerge essentially as soon as the graph becomes temporally connected, \ie the threshold is also $p=3 \log n /n$.
As for pivotal spanners and optimal spanners, both exist \aas at $4 \log n /n$.
For pivotal spanners, this constitutes a sharp threshold. Whether this is also the case for optimal spanners is left open---we conjecture that it is.
These results are summarized in Table~\ref{tab:thresholds-spanner}.

\begin{table}
\centering
\begin{tabular}{|l|c|c|c|}
	\hline
	\textbf{Property} & \textbf{Edge count} & \textbf{Sharp threshold} & \textbf{Reference} \\
	\hline
	Nearly optimal spanner
         & $(2+o(1))n$
            & $3\log n/n$ & \cref{th:main} \\
    \hline
        \multirow{3}*{Optimal spanner} &
        \multirow{3}*{$2n-4$}
            & $\geq 3\log n/n$ & \cref{th:main} \\
            \cline{3-4}
           & & $\leq 4\log n/n$ & \cref{th:opt-span4} \\
            \cline{3-4}
           & & $=4\log n/n$ ? & \cref{conjecture} \\
    \hline
        Pivotal spanner
        & $2n-2$
            & $4\log n/n$ & \cref{th:pivotalSpan} \\
    \hline
\end{tabular}
\caption{Sharp thresholds for spanner properties in random simple temporal graphs.}
\label{tab:thresholds-spanner}
\end{table}

\subsection*{Significance of the results}

In static graphs, no distinction exists between properties (2), (3), and (4). 
If a node $u$ can reach all the others, then by transitivity and symmetry of the reachability relation, every node can reach all the others.
Furthermore, spanning trees of $n-1$~edges exist unconditionally.
In standard Erd{\H{o}}s–R{\'e}nyi graphs, all these properties  occur at the same point of the densification process, namely at~$p=\log n / n$ \cite{EvolutionRandomGraphs1960}.
As for Property (1), it occurs gradually before that point, without obeying a sharp threshold (see, e.g.,~\cite[Section 2.2]{FK16}).
In view of these results, the fact that the temporal versions of the four properties occur at three distinct (and sharp) thresholds is quite remarkable. This can be seen as a fine-grained measure of the discrepancy between static and temporal reachability.

Conceptually, our results regarding the existence of spanners also admit a quite remarkable interpretation.
Recall that the existence of linear spanners, and even subquadratic spanners, is not deterministically guaranteed in temporal graphs, as witnessed by infinite families of adversarial graphs in~\cite{KKK02} and~\cite{AF16} (respectively).
These \emph{obstructions} to the existence of linear (and even sparse) spanners establish a fundamental difference between static and temporal reachability.
Our results show that this difference does not hold probabilistically.
In particular, our result on the existence of nearly optimal spanners, 

\vskip2ex
\noindent
\textbf{\cref{th:main}.}~\textit{There is $\delta = \delta(n) \in o(1)$ such that
$p(n) = 3 \log n / n$ is a sharp threshold for an RSTG to admit a spanner of size at most $(2 + \delta)n$,
}
\vskip2ex

\noindent
combined with the fact that temporal connectivity itself arises at $p=3\log n / n$ (\cref{th:tempConnectivity}), establishes that the obstructions in~\cite{KKK02} and~\cite{AF16}---and in fact, all conceivable obstructions---are statistically insignificant.
In other words, nearly optimal spanners almost surely exist if the graph is temporally connected,
and by analogy to static graphs, the universality of sparse spanners can be recovered at least in a probabilistic sense.

Drawing an analogy to \emph{directed} static graphs, our characterization of a threshold at $p=4 \log n / n$ for the existence of pivotal spanners of size $2n - 2$ is to be contrasted to the unconditional existence of analog constructions in strongly connected static graphs by Kosaraju-Sharir's principle~\cite{sharir1981strong}.
Indeed, the fact that this threshold is different from temporal connectivity (at $p=3 \log n / n$) implies that such constructions are not universal, even in a probabilistic sense.
Finally, whether optimal-size spanners (\ie spanners of size $2n - 4$) are probabilistically universal remains open.
Our upper bound at $p = 4 \log n / n$ proves that this is at least close to being true.

\subsection*{Related models and questions}

The RSTG model studied in this paper is a natural temporal analog of the Erd{\H{o}}s–R{\'e}nyi random graph model.
Another analog considered in the literature is the model in which a random temporal graph is a sequence of independent Erd{\H{o}}s–R{\'e}nyi random graphs \cite{chaintreau2007diameter,clementi2009broadcasting}.
This model and its generalizations have been studied in the literature from different perspectives \cite{clementi08, avin2008explore, grindrod2010evolving, BCF11, zhang2017random, akrida2020fast}.

Our results can be related to a number of known results in the fields of gossip theory, population protocols, edge-ordered graphs, and random weighted graphs.
Reviewing these connections is another contribution of this paper,
and we dedicate \cref{sec:other-models} to it.
The main observations can be summarized as follows.

A classical question in gossip theory asks the time it takes for a rumor to spread among a set of agents through random phone calls. Gossip models where interactions occur at random without repetition correspond to our model. However, in many cases, gossip protocols consider repeated interactions. We show that, even in this case, the early evolution of reachability is sufficiently close to RSTGs for the known thresholds to be essentially equivalent.
A similar observation holds for some models of population protocols (see, e.g.,~\cite{DBLP:conf/podc/AngluinADFP04}), where the pattern of interactions is specified by a sequential scheduler.
Unless this scheduler is adversarial, it is generally assumed, for analysis, that the interactions are chosen uniformly at random with repetitions. 
Surprisingly, few connections have been made between gossip theory and population protocols, some analyses being re-discovered from time to time in both areas.
Following the notations in~\cite{van2017reachability}, we will refer to the setting with repeated interactions as the \texttt{ANY} model and to the one without repetition as the \texttt{CO} (``call once'') model, irrespective of the areas in which they were considered. 

Early works in the \texttt{ANY} model appear in a sequence of three papers from the '70s
\cite{moon1972random, boyd1979random, haigh1981random},
all under the title ``\emph{Random exchanges of information}'' but with different authors.
These papers are concerned with the time it takes for a fixed agent to receive a potential rumor from every other agent, which by symmetry corresponds to all the agents receiving a fixed single rumor---our third property.
They also consider the time it takes for all agents to receive each other's rumors---our fourth property (temporal connectivity). 
Asymptotics for the \emph{expected} number of interactions are given in these papers, namely $n \log n$ for the former (which indeed corresponds to the number of edges in an RSTG with $p=2 \log n / n$), and $1.5 n \log n$ for the latter (corresponding to $p=3 \log n / n$).
In more recent works, the former was shown to be concentrated around its expected value (the best known result appears to be by~\textcite{Emmanuelle}).
As for temporal connectivity in the \texttt{ANY} model, it is already known that it does not exceed its expected value \aas (see~\textcite{burman2021time} for the currently best bound\footnote{\cite{burman2021time} claims a weaker result (see Lemma~2.9 therein), but an intermediate step of the proof directly implies the upper bound of $(1.5+o(1))n\log n$.}).
In this paper, we show how our results can be translated to the \ANY{} model. Doing so, we obtain concentration results (\ie sharp thresholds) for these two properties, and beyond, for all others. 
Allowing for repetitions in the \ANY{} model actually makes things easier to handle than in the \CO{} model (the one equivalent to RSTGs), where handling dependencies between past and future calls requires additional efforts. According to~\cite{van2017reachability}, not much was known so far in that model, for the same reason. Therefore, our results on RSTGs fill a gap, by showing, among other things, that information propagates at essentially the same speed with or without repetitions.

The fact that RSTGs relate to processes in gossip theory and population protocols is not surprising. Less intuitively,
the gradual levels of temporal reachability that we observe in this work turn out to be also related to percolation processes in randomly edge-weighted complete graphs, as studied by \textcite{Janson99}, provided one interprets the edge weights as waiting times.
This connection is intriguing, because Janson's model is not dynamic: the availability of the edges and the value of weights are invariant in time.
Thus, in particular, the time required to cross an edge (\ie the edge weight) is independent of all other
factors, whereas in temporal graphs and gossip models two temporal paths, in general, need to wait different amounts of time to cross the same edge.
Yet, some of our thresholds for temporal reachability are similar to percolation thresholds in this model;
we discuss the reasons in \cref{sec:first-passage-percolation}.
This agreement indicates that the essential features defining the connectivity properties of random graphs (weighted or temporal) coincide for a very wide range of models, even though the required proof techniques differ between them.

Finally, while the model of edge-ordered graphs is closest to RSTGs,
research here has focused more on longest monotone paths and walks than on characterizing the reachability that results from monotone paths.
Yet, the known results in randomly edge-ordered graphs translate directly to RSTGs, and vice versa.

\subsection{Overview of the techniques}

The following content is intended to help the reader navigate through the ideas and techniques used in this paper.
In general, in order to prove that an RSTG with edge probability~$p$
has (or lacks) some property,
we instead consider an RSTG with edge probability~$1$ and each edge label randomly chosen from~$[0, 1]$.
Then, we study the properties of the temporal subgraph induced by all the edges whose labels are at most $p$.

A preliminary result in \cref{sec:2hop} establishes that \aas every vertex can reach every other vertex
via a temporal path of length at most $2$ when $p\geq \log n / \sqrt{n}$.
This result is much weaker than the subsequent ones, but serves as a warm-up to the reader and allows us to subsequently assume that all relevant labels are contained in $[0, t]$ for some $t=o(1)$.

A concept that is quite versatile in the paper is that of a \emph{foremost tree}, \ie a tree consisting of gluing together a (prefix-stable) set of earliest arrival paths from a given vertex, called a temporal source, to all others. Foremost trees were considered for different purposes in~\cite{BFJ03} and~\cite{LL14}, and a similar concept also appears in~\cite{kleinberg2006algorithm} (Problem 4.18).
\cref{sec:foremostTree} offers a general study of foremost trees in RSTGs. In particular,
we define a suitable martingale for estimating the time by which $k$ vertices have been reached from a given vertex $u$,
which allows us to apply Azuma's inequality and derive concentration results.
Precisely, we consider the sequence of waiting times for the respective next vertex to be reached from the source $u$
and turn this into a martingale by subtracting the respective expected values.
As the probability distribution of these waiting times has a one-sided long tail, 
we are able to achieve improved concentration bounds by analyzing capped versions of these variables (Lemma 4.6).
This technique allows us to obtain a sharp threshold for the appearance of the first temporal source.
On the other hand, these caps are chosen such that \aas all elements of the sequence agree with their uncapped versions.
By this technique, we also obtain sharp thresholds for the event that a fixed vertex reaches a particular number of other vertices (\cref{th:reachability}), which
yields, in turn, sharp thresholds for temporal reachability between
two fixed vertices and for a fixed vertex being a temporal source.

In order to characterize the threshold for temporal connectivity, 
we start by considering the moment when the first
temporal sink arises (same as for temporal sources, by symmetry).
Naively, one might be tempted to wait for the same amount of time subsequently for that vertex to reach all others back as a temporal source, which gives temporal connectivity by the pivoting principle (more on this below). However, this implies waiting for $4 \log n / n$ overall, whereas the actual threshold for temporal connectivity occurs at $3 \log n / n$ and requires a more general approach.
Starting with the lower bound, we show that once the first temporal sink has appeared, an additional $\log n / n$ wait is needed for all the other vertices to get at least one more edge,
which is necessary for them to become temporal sinks in turn (\cref{lm:tconnLower}). The fact that temporal connectivity is equivalent to having all the vertices being temporal sinks (or sources, for that matter) implies, in turn, that $3 \log n / n$ is indeed a lower bound.
As for the upper bound, we use the fact that the majority of vertices become temporal sinks at $2 \log n / n$ (soon after the first temporal sink appears),
after which an additional $\log n / n$ wait guarantees that
\aas each remaining vertex
becomes reachable from everywhere
via at least one of these sinks (\cref{lem:connUpperBound}).

To establish the existence of small temporal spanners in \cref{sec:temporalSpanners}, we give explicit constructions.
For pivotal spanners, we require a pivot vertex that is 
a temporal sink before $2 \log n / n$ and a temporal source after that, but before $4 \log n / n$. The lower bounds are based on the above ones for reachability.
Truly optimal spanners are built in a similar way, but instead of a single pivot, they rely on the concept of a pivotal square whose existence is guaranteed by time $4 \log n / n$.
The result that $3 \log n / n$ is a sharp threshold for the existence of nearly optimal spanners requires a finer understanding of the height of foremost trees.
To this end, we compute inductively the approximate distribution of vertex heights in a foremost tree, which implies a logarithmic upper bound on the number of edges in the foremost path
between any two vertices (\cref{sec:foremostTreeHeight}).
Then, we divide the time into three intervals of length $\log n / n$ each.
In the first interval, we find a temporal tree through which most vertices can reach a vertex that will act almost as a pivot.
Analogously, in the last time interval, we build a temporal tree through which this vertex can reach most vertices.
The remaining vertices are individually connected to the pivot by earliest and latest temporal paths whose edges are distinct from the edges of the trees.
Finally, we separately connect all pairs of vertices that still cannot reach each other via the pivot. The contribution of these additional paths to the spanners is shown to be asymptotically negligible, which implies an overall size of $(2 + o(1))n$ edges.

\subsection{Organization of the document}

The main definitions are given in Section~\ref{sec:preliminaries}, including basic concepts and notations in temporal graphs.
Section~\ref{sec:2hop} presents an analysis that obtains weaker bounds than our actual results using simpler arguments; while proposed as a warm-up, these bounds are also utilized in subsequent analysis. 
Section~\ref{sec:foremostTree} develops our main tools, by describing an algorithm that grows a \emph{foremost tree} (\ie a tree of time-optimal temporal paths) in a given temporal graph and analyzing the execution of this algorithm on a typical RSTG.
Section~\ref{sec:tempProp} applies the tools from the previous section to obtain the main results, namely the claimed thresholds in RSTGs.
Section~\ref{sec:other-models} describes the adaptation of our analyses to models coming from gossip theory and population protocols.
We also explain in detail how our results strengthen known results in these fields.
Finally, Section~\ref{sec:conclusion} concludes with some remarks and open questions.

\section{Preliminaries}
\label{sec:preliminaries}

\subsection{Temporal graphs}

A \emph{temporal graph} $\G$ is a pair $(G, \lambda)$, where
$G=(V,E)$ is a simple undirected graph and $\lambda$ is a function that
assigns to every edge $e$ of $G$ a finite set of elements from some totally ordered set $\mathbb{T}$. 
The graph $G$ is called the \emph{underlying graph} of $\G$, the elements of $\lambda(e)$ are called \emph{time labels of $e$}, and the pairs $(e, t)$, where $t \in \lambda(e)$ are called \emph{time-edges}.
The temporal graph $\G$ is \emph{simple} if every edge of $G$ is assigned exactly one time label, \ie $|\lambda(e)| = 1$ for every $e \in E$. 
We will mostly focus on simple temporal graphs, in which case
with a slight abuse of notation we write $\lambda(e)$ to denote the unique time label of an edge $e \in E$
and refer to the time-edge $(e, \lambda(e))$ simply as edge $e$.

A \emph{temporal $(u,v)$-path} or a \emph{temporal path} from $u$ to $v$ in $\G$ is a path $u = u_0, u_1, \ldots, u_{\ell} = v$ in $G$ such that
$\lambda_1 \le \lambda_2 \le \cdots \le \lambda_{\ell}$, where 
$\lambda_i \in \lambda(u_{i-1}u_i)$.
The labels $\lambda_1$ and $\lambda_{\ell}$ are called the \emph{departure time} and \emph{arrival time} of the temporal path, respectively.
A temporal $(u, v)$-path is called a \emph{foremost} $(u, v)$-path if it 
has the earliest arrival time among all temporal $(u, v)$-paths. 
Symmetrically, a temporal $(u, v)$-path is called a \emph{hindmost} $(u, v)$-path if it 
has the latest departure time among all temporal $(u, v)$-paths. 
A vertex $v \in V$ is called a \emph{temporal source of $\G$} if every other vertex
can be reached from $v$ by a temporal path.
Similarly, vertex $v$ is called a \emph{temporal sink of $\G$} if every other vertex can reach $v$ by a temporal path.
The temporal graph $\G$ is \emph{temporally connected} if each vertex can reach every
other vertex by a temporal path.
A temporal graph $\G'=(G', \lambda')$ is a \emph{temporal subgraph}
of $\G$ if $G'$ is a subgraph of $G$ and $\lambda'(e) \subseteq \lambda(e)$ for every edge $e$ of $G'$.
Furthermore, if $V(\G') = V(\G)$ and $\G'$ is temporally connected, then $\G'$ is called a \emph{(temporal) spanner} of $\G$. All the temporal graphs we consider in this work are simple. As a result, spanners are uniquely determined by the set of underlying edges that they contain.

It is a straightforward consequence of a classical result in gossiping theory (see, e.g., \cite{bollobas2006art}) that any
spanner of an $n$-vertex simple temporal graph has at least $2n-4$ edges. (The same holds more generally for time-edges in non-simple temporal graphs.) 
For this reason, we will refer to spanners with exactly $2n-4$ edges as \emph{optimal}.
Another important type of temporal spanners is \emph{pivotal} temporal spanners. Such a spanner
is the union of two disjoint spanning trees rooted at the same vertex $v$, called \emph{pivot}, 
such that for some time $t$ all vertices can reach $v$ in one of the trees before time $t$ and $v$ can reach all vertices in the other tree after time $t$.
Clearly, a pivotal spanner contains $2n-2$ edges.
Finally, we define the \emph{restriction of $\G = (G, \lambda)$ to a time interval~$[a, b]$},
denoted by $\G_{[a,b]}$, as the temporal graph $\G' = (G', \lambda')$
where $\lambda'(e) = \lambda(e) \cap [a,b]$
and $G' = (V, \{e \in E \mid \lambda'(e) \neq \emptyset\})$.

\subsection{Random simple temporal graphs (RSTGs)}

The most basic and commonly studied model of random static graphs is the Erd{\H{o}}s–R{\'e}nyi model
denoted by $G_{n,p}$, in which there are $n$ vertices and every two of them are connected by an edge 
independently with probability $p \in [0,1]$.
A possible way of turning $G_{n,p}$ into a model of random \emph{temporal} graphs is by choosing uniformly at random a total order on the set of edges.
More specifically, a random temporal graph $(G,\lambda)$ can be obtained by first drawing the underlying graph $G$ from $G_{n,p}$ and then 
drawing $\lambda$ uniformly at random from all bijections $E(G) \rightarrow \{ 1, 2, \ldots, |E(G)| \}$.
For technical convenience, however, we work with a slightly different but equivalent model which we denote by $\F_{n,p}$.
In this model, the underlying graph is drawn from $G_{n,p}$ as before, but the labeling function $\lambda$ now maps every edge of $E(G)$ to an independent and uniformly distributed label in the real interval~$[0, 1]$.
Since with probability~1 no two time labels are equal, this induces a total order on the
edges, and by symmetry, all orders are equally likely; thus $\F_{n,p}$ is equivalent to the above model
of random edge orderings of $G_{n,p}$. (This equivalence has been used in some recent works on edge-ordered graphs \cite{angel2020long,LL14}.)
We refer to such a graph as a random simple temporal graph (RSTG), or just a random temporal graph, which is simple by default.
Throughout the paper we will assume, without loss of generality, that in the graphs from $\F_{n,p}$
that we work with all edges have pairwise different labels.

Instead of working with $\F_{n,p}$ directly,
it will often be convenient to first draw a temporal graph $\G = (G, \lambda)$ from $\F_{n,1}$ (resulting in $G$ being a complete graph),
and to then consider $\G' = (G', \lambda') =  (G, \lambda)_{[0, p]}$ in which all edges with time labels greater than $p$ are deleted.
Observe that $G' \sim G_{n,p}$ and that each edge label $\lambda'(e)$ is uniformly distributed on $[0,p]$.
In other words, $\G'$ is distributed according to $\F_{n,p}$ up to multiplying all labels by $\frac{1}{p}$. Furthermore, since the temporal properties that we study depend only on the relative order of edge labels and not on their absolute values, without loss of generality, we will omit the multiplicative factor $\frac{1}{p}$ and will work directly with $\G'$. 
Similarly, for any $0 \leq a \leq b \leq 1$, up to rescaling $x \mapsto \frac{x-a}{b-a}$, the graph $(G, \lambda)_{[a, b]}$ is distributed 
according to $\F_{n,p}$, where $p = b-a$.

\subsection{Degrees of reachability}

A graph property is said to hold asymptotically almost surely (a.a.s.) if the probability that it is satisfied converges to $1$ as $n$ goes to infinity.
A function $p = p(n)$ is a \emph{sharp threshold} for a temporal graph property $\Pp$ if for every $\varepsilon > 0$ a random temporal graph from $\F_{n,(1-\varepsilon)p}$ does not have property $\Pp$ \aas, while a random temporal graph from $\F_{n,(1+\varepsilon)p}$ has $\Pp$ a.a.s.

We study the following fundamental properties related to temporal reachability:
\begin{enumerate}
		\item \textbf{Point-to-point Reachability}. The property that a fixed pair $(u, v)$ of vertices has a temporal path from $u$ to $v$.

		\item \textbf{First Temporal Source}. The property that $\G$ contains at least one temporal source.
		
		\item \textbf{Temporal Source}. The property that a fixed vertex $v$ is a temporal source, \ie every vertex in $\G$ can be reached from $v$ by a temporal path.
		
		\item \textbf{Temporal Connectivity}. The property that $\G$ is temporally connected, which is equivalent to the property that all vertices in $\G$ are temporal sources.
		
		\item \textbf{Nearly Optimal Temporal Spanner}. The property that $\G$ contains a temporal spanner with $(2+o(1))n$ time-edges.
		
		\item \textbf{Optimal Temporal Spanner}. The property that $\G$ contains a temporal spanner with $2n - 4$ time-edges.

		\item \textbf{Pivotal Temporal Spanner}. The property that $\G$ contains a pivotal temporal spanner.
		
\end{enumerate}

\noindent
In Section~\ref{sec:tempProp}, we establish sharp thresholds for all properties in RSTGs (see \cref{tab:thresholds}) except for the sixth, for which we prove an upper bound that we conjecture to be a sharp threshold as well.

\section{Warm-up: 2-hop approach}\label{sec:2hop}

This section presents a simple argument to derive upper bounds on temporal source related properties. 
The main statement says that for any $p \geq 3\sqrt{\log{n}/n}$ a fixed vertex $v$ is a temporal source 
in a random temporal graph from $\F_{n,p}$ {a.a.s.} The simplicity of this approach comes from the restriction
that we consider only temporal paths of length~$2$. While these bounds are far from optimal, they turn out to be useful for subsequent analyses.
Note that here and in the following, we often implicitly assume $n$~to be at least an appropriate constant.

\begin{lemma}\label{thm:2hop-result}
	Let $\alpha \geq 3$ and let $p = \alpha\sqrt{\log{n}/n}$.
	Then an arbitrary vertex of $(G, \lambda) \sim \F_{n,p}$ is a temporal source 
	with probability at least $1 - n^{-\alpha^2/4 + 1}$.
\end{lemma}
\begin{proof}
	Let $x$ be an arbitrary vertex in $(G, \lambda)$.
	For two distinct vertices $y, z$ of $(G, \lambda)$ that are also different from $x$,
	we denote by $R_z$ the event that $x$ reaches $z$, and by $S_{yz}$ we denote the event that 
	$x$ reaches $z$ in exactly two steps via $y$. Notice that $\prob\left[\overline{R_z}\right]$ is the same for every 
	$z \in V(G) \setminus \{x\}$, and we denote this probability by $p_1$. 
	Similarly, $\prob[S_{yz}]$ is the same for all pairs $z,y \in V(G) \setminus \{x\}$, $z \neq y$, and is equal to $p^2/2$. 
	Hence, denoting by $\gamma$ the probability that $x$ is a temporal source in $(G, \lambda)$ and by using
	the union bound, we have
	\[
		\gamma = 1 - \prob \Big( \bigcup_{z\neq x} \overline{R_z} \Big) 
		\geq 1 - \sum_{z \neq x} \prob\left[\overline{R_z}\right]
		= 1 - (n-1)p_1.
	\]
	Furthermore, for every $n \geq 4$, we have
	\begin{equation*}
	\begin{split}
		p_1 = \prob\left[\overline{R_z}\right]
		& \leq \prob\Bigl[\bigcap_{y\neq x,z} \overline{S_{yz}}\Bigr] 
		= \prod_{y\neq x,z} \prob\left[\overline{S_{yz}}\right]
		= \prob\left[\overline{S_{yz}}\right]^{n-2} \\
		& = (1- p^2/2)^{n-2} \leq e^{-\frac{p^2(n-2)}{2}}
		= \left( \frac{1}{n} \right)^{\alpha^2 \cdot \frac{n-2}{2n}}
		\leq \left( \frac{1}{n} \right)^{\frac{\alpha^2}{4}},
	\end{split}
	\end{equation*}
	where we used the fact that all $S_{yz}$ are independent.
	Consequently, we derive the desired conclusion
	\[
		\gamma \geq 1 - (n-1) n^{-\alpha^2/4} \geq  1 - n^{-\alpha^2/4 + 1}.
              \]
\end{proof}

If $\alpha=3$ in the above lemma, we obtain that for $p = 3\sqrt{\log{n}/n}$ with probability at least $1 - 1/n^{5/4}$ an arbitrary vertex in $(G, \lambda) \sim \F_{n,p}$ is a temporal source. Similarly, if $\alpha=4$, by applying the union bound, we obtain that for $p = 4\sqrt{\log{n}/n}$ with probability at least $1 - 1/n^2$ \emph{every} vertex in $(G, \lambda) \sim \F_{n,p}$ is a temporal source, \ie $(G, \lambda)$ is temporally connected. However, in the subsequent analysis we will use the following corollary of the lemma obtained by taking $\alpha = \sqrt{\log n}$ and applying the union bound.

\begin{corollary}\label{cor:2hopTempConn}
	Let $p = \frac{\log{n}}{\sqrt{n}}$. Then, $(G, \lambda) \sim \F_{n,p}$ is temporally connected with probability at least $1-n^{-\frac{\log{n}}{4} + 2}$.
\end{corollary}

\section{Foremost tree evolution}\label{sec:foremostTree}

Let $\G = (G, \lambda)$ be a random temporal graph from $\F_{n,1}$.
Consider the probability that a fixed vertex $v$ reaches another fixed vertex
$u$ in $\F_{n,p}$.
This is equal to the probability that 
the temporal subgraph $\G_{[0, p]}$ contains
a temporal $(v,u)$-path $P$.
The latter is equivalent to the fact
that the arrival time of $P$ in $\G$ is at most $p$.
Therefore, the estimation of the parameter $p$ for temporal reachability
from $v$ to $u$ can be reduced to the estimation of the arrival time of a
foremost temporal $(v,u)$-path in $\G$.
In case of the Temporal Source property, we are interested in the smallest value of $p$ such that 
a given vertex $v$ is a temporal source in $\G_{[0, p]}$, or equivalently,
that any vertex in $\G$ can be reached from $v$ by time $p$ (notice that since $\G$ is a complete temporal graph, every vertex in $\G$ is reachable from $v$, e.g., by a single-edge path). A minimal temporal subgraph 
that preserves foremost reachabilities from $v$ to all other vertices reachable from $v$ is called 
\emph{foremost tree for $v$} (formally defined below), and  we will be interested in estimating the smallest $p$ such that $\G$ contains a foremost tree for $v$ with time labels of its edges not exceeding $p$.
We proceed by defining formally the relevant notions.

Let $(T, \lambda')$ be a temporal graph, where $T$ is a tree, and let $u$ be a vertex in $T$.
If $u$ is a temporal source (resp.\ temporal sink) in $(T, \lambda')$, then we say
that $(T, \lambda')$ is an 
\emph{increasing temporal tree} (resp.\ \emph{decreasing temporal tree}) rooted at $u$.
A temporal subgraph $(T, \lambda')$  of $\G$ is called a \emph{partial foremost tree for $u$}
in $\G$ if $(T, \lambda')$ is an increasing temporal tree rooted at $u$ and
for every vertex $v$ of $T$, the temporal $(u,v)$-path in $(T, \lambda')$ is a foremost
$(u,v)$-path in $\G$.
Symmetrically, a temporal subgraph $(T, \lambda')$  of $\G$ is called a \emph{partial hindmost tree for $u$} in $\G$ if 
$(T, \lambda')$ is a decreasing temporal tree rooted at $u$ and
for every vertex $v$ of $T$, the temporal $(v,u)$-path in $(T, \lambda')$ is a hindmost
$(v,u)$-path in $\G$.
A partial foremost tree (resp.\ partial hindmost tree) $(T, \lambda')$ for $u$ in $\G$ 
is called \emph{foremost tree} (resp.\ \emph{hindmost tree}) for $u$ in $\G$, 
if $T$ contains all vertices that are reachable from $u$ (resp.\ all vertices that can reach $u$) in $\G$.

Let $(G, \lambda)$ be a simple temporal graph with time labels from an interval $[a,b]$.
Let $\lambda^r_{[a,b]} : E(G) \rightarrow [a,b]$ be the mapping that mirrors the time labels at the middle point of $[a,b]$, \ie $\lambda^r_{[a,b]}(e) := b - \lambda(e) + a$. Then $(G, \lambda^r_{[a,b]})$ is also a simple temporal graph with time labels in $[a,b]$, but the order of time labels is reversed with respect to $(G, \lambda)$.
In particular, any temporal $(u,v)$-path in $(G, \lambda)$ is a temporal $(v,u)$-path in $(G, \lambda^r_{[a,b]})$, and an increasing (resp.\ decreasing) temporal tree rooted at some vertex $u$ in $(G, \lambda)$ is a decreasing (resp.\ increasing) temporal tree rooted at $u$  in $(G, \lambda^r_{[a,b]})$.
Also, any foremost $(u,v)$-path in $(G, \lambda)$ becomes a hindmost $(v,u)$-path
in $(G, \lambda^r_{[a,b]})$ and vice versa; hence any (partial) foremost tree for $u$ in $(G, \lambda)$ is a (partial) hindmost tree for $u$ in $(G, \lambda^r_{[a,b]})$ and vice versa.
Furthermore, since $\lambda^r_{[a,b]}(e)$ is a bijection, if $\G \sim \F_{n,1}$ and $(G', \lambda) = \G_{[a,b]}$, then the graphs 
$(G', \lambda)$ and $(G', \lambda^r_{[a,b]})$ have the same distribution.
Due to this relation, in the present section, we will study only foremost paths and trees, but analogous results also hold for hindmost paths and trees. We will state specific results for hindmost paths and trees if they are needed subsequently.

The main purpose of the present section is to analyze temporal and structural properties of a typical foremost tree in $\F_{n,1}$.
In \cref{sec:foremostAlg} we present an algorithm for constructing a foremost tree for a given temporal source.
In \cref{sec:foremostGrowth} we analyze the execution of the algorithm on a random instance from $\F_{n,1}$ to estimate the speed of growth of the foremost tree.
We utilize these results in \cref{sec:tempProp} for obtaining sharp thresholds.

\subsection{Foremost tree algorithm}
\label{sec:foremostAlg}

We start by presenting an algorithm that, given a temporal graph and a source%
\footnote{For the sake of simplicity, we assume that the input vertex is a source; however, after trivial adjustments in the pseudocode of the algorithm will also work for any input vertex} vertex $v$ 
in the graph constructs a foremost tree for $v$.
The algorithm is similar to Prim's algorithm for a minimum spanning tree in static graphs with the only difference being
that the next edge to be added to the tree is chosen as the edge with minimum label
\emph{among those that extend the current tree to an increasing temporal tree}.
Foremost trees have been previously studied by \textcite{BFJ03}, as well as by \textcite{LL14}. A similar idea also appears implicitly in Problem 4.18 of the classic algorithms book by \textcite{kleinberg2006algorithm}.

\begin{algorithm}[H]
	\caption{\textsc{Foremost Tree}}

	\begin{algorithmic}[1]
		\Require{Simple temporal graph $(G, \lambda)$; temporal source $v$ in $(G, \lambda)$.}
		\Ensure{Foremost tree for $v$.}
			
		\State{$T_0 := (\{v\}, \emptyset)$}
			
		\For{$k := 1$ to $n-1$}
			\State{Let $S_k$ be the set of edges in $G$ with one endpoint in $V(T_{k-1})$ and the other
			endpoint \newline\hspace*{1cm} in $V(G) \setminus V(T_{k-1})$}
			\State{$e_k := \argmin \{ \lambda(e) \mid e \in S_k \text{ and } T_{k-1} \cup \{e \} \label{line:ek}
			\text{ is an increasing temporal tree rooted} \newline\hspace*{1cm} \text{ at } v \}$}
			\State{$T_k := T_{k-1} \cup \{ e_k \}$}
		\EndFor

		\noindent	
		\Return{$(T_{n-1}, \lambda')$, where $\lambda'$ is the restriction of $\lambda$ to the edges of 
		$T_{n-1}$}
	\end{algorithmic}
	\label{alg:foremostTree}
\end{algorithm}

In the next lemma we prove the correctness of the algorithm and the easy but
key fact that the time labels of the tree edges monotonically increase in the order in which they are added.

\begin{lemma}\label{lem:foremostTree}
	Let $(G, \lambda)$ be a simple temporal graph and $v$ be a temporal source in $(G, \lambda)$. Then
	\begin{enumerate}[(i)]
		\item \cref{alg:foremostTree} constructs a foremost tree for $v$ in $(G, \lambda)$;
		\item \label{lambda_inequality} $\lambda(e_1) \leq \lambda(e_2) \leq \cdots \leq \lambda(e_{n-1})$.
	\end{enumerate}
\end{lemma}
\begin{proof}
	To prove the first claim of the lemma, we will show by induction on~$k$ that 
	\begin{enumerate}
		\item[(a)] $T_k$ is well-defined, \ie the set from which $e_k$ is chosen is non-empty;
		
		\item[(b)] for every $u \in V(T_k)$ the temporal $(v,u)$-path in $T_k$ is a foremost $(v,u)$-path in $(G, \lambda)$.
	\end{enumerate}
	The statement is obvious for $k=0$. Let $1 \leq k \leq n-1$ and assume the statement holds for~$k-1$. We will show that it also holds for~$k$.
	
	Let $u$ be an arbitrary vertex in $V(G) \setminus V(T_{k-1})$, and let $P$ be a foremost $(v,u)$-path in $(G, \lambda)$.
	Let $e = ab$ be the first edge of $P$ with one endpoint, say $a$, in $T_{k-1}$ and the other endpoint $b$ not in $T_{k-1}$.
	Let also $P'$ be the foremost $(v,a)$-path in $T_{k-1}$. Since $P'$ is foremost in $(G, \lambda)$, 
	the arrival time of $P'$ is not more than the arrival time of the $(v,a)$-subpath of $P$. Consequently, the arrival time of $P'$ is 
	not larger than $\lambda(e)$, and therefore $T_{k-1} \cup \{ e \}$ is an increasing temporal tree, which proves part (a) of the statement.
	
	Let now $e_k = ab$ be the edge added to $T_{k-1}$ to form $T_k$, where $a \in V(T_{k-1})$ and $b \notin V(T_{k-1})$. 
	Taking into account the induction hypothesis, to prove part (b) of the statement, we only need to show that the temporal 
	$(v,b)$-path in $T_k$ is a foremost $(v,b)$-path in $(G,\lambda)$. 
	Suppose it is not, and let $P$ be a foremost $(v,b)$-path in $(G,\lambda)$. Let $e'$ be the first edge of $P$
	with one endpoint in $V(T_{k-1})$ and the other endpoint in $V(T_{k})$. By the proof of part (a), we know that $T_{k-1} \cup \{ e' \}$
	is an increasing temporal tree. Furthermore, since the arrival time of $P$ is less than $\lambda(e_k)$, 
	we have that $\lambda(e') < \lambda(e_k)$. But this contradicts the choice of $e_k$.
	
	To prove the second claim of the lemma, assume that it does not hold and let $k \geq 2$ be the minimum index such that 
	$\lambda(e_k) < \lambda(e_{k-1})$. Let $e_k = ab$, where $a \in V(T_{k-1})$ and $b \notin V(T_{k-1})$, and
	let $e_i$ be the last edge of the $(v,a)$-path in $T_k$. Clearly $\lambda(e_i) \leq \lambda(e_k)$ because $T_k$ 
	is increasing. Hence, there exists $j$ with $i < j < k$, such that $\lambda(e_{j-1}) \leq \lambda(e_{k})  < \lambda(e_{j})$.
	Since $i \leq j-1$, the edge $e_i$ belongs to $T_{j-1}$, and therefore $e_k$ can extend $T_{j-1}$ to an increasing temporal tree.
	But this contradicts the choice of $e_j$ at the $j$-th iteration of the algorithm, as $\lambda(e_k) < \lambda(e_j)$.
\end{proof}

\subsection{Foremost tree growth}
\label{sec:foremostGrowth}

The main goal of this section is to estimate the time by which a typical foremost tree in $\F_{n,1}$ acquires a given number of vertices.
In this subsection, we always assume~$n$ to be at least a sufficiently large constant.
To conduct our analysis, we will consider the execution of \cref{alg:foremostTree} on a complete random temporal graph 
$(G, \lambda) \sim \F_{n,1}$ from some fixed vertex $v$ as a random process that reveals the edges 
of the resulting foremost tree for $v$ one by one in the order in which they are added. 
We define random variables $Y^v_0 := 0$ and  $Y^v_k := \lambda(e^v_k)$, $k \in [n-1]$, where, following the algorithm, $T^v_0 :=(\{v\}, \emptyset)$ 
and for every $k \in [n-1]$
\begin{equation*}
	\begin{split}
		S^v_k &:= V(T^v_{k-1}) \times \big( V(G) \setminus V(T^v_{k-1}) \big),\\
		e^v_k &:= \arg\min \{ \lambda(e) ~|~ e \in S^v_k \text{ and } T^v_{k-1} \cup \{e \} \text{ is an increasing temporal tree} \}, \\
		T^v_k &:= T^v_{k-1} \cup \{ e^v_k \}.
	\end{split}
\end{equation*}

By definition, $Y^v_k$ is the earliest time when the foremost tree for $v$ contains exactly $k$ edges, or equivalently the earliest time by which $v$ can reach $k+1$ vertices (itself included). 
We note that since $(G, \lambda)$ is temporally connected, \cref{lem:foremostTree} implies that all $Y^v_k$, $k \in [n-1]$, are finite.
For $k \in [n-1]$, let $X^v_k$ be a random variable equal to $Y^v_k - Y^v_{k-1}$, 
\ie to the \emph{waiting time} between the edges $e^v_{k-1}$ and $e^v_k$. Clearly, we have 
\[
	Y^v_k = \sum_{i=1}^{k} X^v_i
\]
for every $k \in [n-1]$.

The main objects of our analysis are the random variables $X^v_1, X^v_2, \ldots, X^v_{n-1}$ and $Y^v_1, Y^v_2, \ldots, Y^v_{n-1}$.
We will also study the behavior of their \emph{truncated} versions,
which are convenient in the applications:
for $k \in [n-1]$, let $\capped{X}^v_k := \min \{ X^v_k, c_k \}$ and 
$\capped{Y}^v_k := \sum_{i=1}^{k} \capped{X}^v_i$, where
\[
	c_{k} := \frac{2 \log{(\min\{ k, n-k\})} + \log \log{n}}{k(n-k)}.
\]
The values of $c_k$ are chosen in such a way that on the one hand they are sufficiently small
to ensure the applicability of Azuma's inequality (\cref{the:Azuma});
and on the other hand they are large enough to guarantee that the truncated variables coincide
with their original versions a.a.s.
The numerical conditions are formalized in the following lemma.

\begin{lemma}[Properties of $c_k$]
	\label{lm:XCapAAS_sum}
	For any sufficiently large $n$, we have
	\begin{enumerate}[(i)]
		\item $\sum_{i=1}^{n-1} c_i^2 \leq \frac{64 (\log \log{n})^2}{n^2}$; \label{ck_bound}
                \item$\sum_{i=1}^{n-1}(1-c_k)^{k(n-k)}\leq 4/\log{n}$
                        \label{capped_eq_uncapped_sum}
	\end{enumerate}
\end{lemma}
\begin{proof}
	We start by proving the first part of the lemma.
	\begin{equation*}
		\begin{split}
			\sum_{i=1}^{n-1} c_i^2  & = \sum_{i=1}^{n-1} 
			\frac{\big(2 \log{(\min\{ i, n-i\})} + \log \log{n}\big)^2}{(i(n-i))^2}
			\leq \sum_{i=1}^{n-1} \frac{\big(2  \log{(\min\{ i, n-i\})} \cdot \log \log{n}\big)^2}{(i(n-i))^2} 
			\\ &\leq 4 (\log\log{n})^2 \left( \sum_{i=1}^{\lfloor n/2 \rfloor} 
			\frac{(\log i)^2}{i^2(n/2)^2} + 
			\sum_{i=\lfloor n/2 \rfloor+1}^{n-1} \frac{(\log(n-i))^2}{(n/2)^2(n-i)^2} \right)
			\\ &\leq \frac{32 (\log \log{n})^2}{n^2} \sum_{i=1}^{\infty}\frac{(\log{i})^2}{i^2} \leq
			\frac{64 (\log \log{n})^2}{n^2},
		\end{split}
	\end{equation*}	
	where the first inequality holds for all sufficiently large $n$.
	
	To prove the second part, we observe that 
	\begin{equation}\label{eq:prob_capped_uncapped_sum}
	\begin{split}
                (1 - c_k)^{k(n-k)} \leq e^{-k(n-k)c_{k}} 
		= \frac{1}{(\min\{ k, n-k\})^2 \log{n}},
	\end{split}
	\end{equation}
	and therefore
	\begin{align*}
			\sum_{k=1}^{n-1}
                        (1-c_k)^{k(n-k)} & \leq 
			\sum_{k=1}^{n-1} \frac{1}{(\min\{ k, n-k\})^2 \log n} \\ &=
			\frac{1}{\log{n}} \left( \sum_{k=1}^{\lfloor n/2 \rfloor} \frac{1}{k^2} +
			\sum_{k=\lfloor n/2 \rfloor+1}^{n-1} \frac{1}{(n-k)^2} \right) \\
			&\leq \frac{2}{\log{n}} \sum_{k=1}^{\infty}\frac{1}{k^2} \leq
			\frac{4}{\log{n}}.
	\end{align*}
\end{proof}

We will now estimate the expected time between the moments of 
exposing two consecutive edges of the foremost tree.
More specifically, we will show that 
$X_k^v$ and $\capped{X}_k^v$ coincide for a given vertex a.a.s., and
bound the expected values of $X_k^v$ and $\capped{X}^v_k$ 
given the information revealed by the process in the first $k-1$ steps.
For every $k \in [n-1]$, let $\A^v_k$ be the $\sigma$-algebra generated 
by the information revealed at the first $k$ iterations of the algorithm starting at $v$, \ie by the knowledge of the first 
$k$ edges $e^v_1, e^v_2, \ldots, e^v_k$ and their time labels $Y^v_1, Y^v_2, \ldots, Y^v_k$. 
Let also $\A^v_0$ be the trivial $\sigma$-algebra.
Then, we have the following

\begin{lemma}\label{lem:ExpOfXk}
For sufficiently large $n$, for a vertex $v$, with probability at least $1 - 4/\log{n}$ the equality $\capped{X}^v_k=X^v_k$ holds for every $k \in [n-1]$.

Further,
	for a vertex $v$ and  every $k \in [n-1]$, we have
	\begin{flalign*}
		 \text{(i)} && \frac{1-Y^v_{k-1}}{k(n-k) + 1}  &\leq  \expect[X^v_k \mid \A^v_{k-1}]  \leq  \frac{1}{k(n-k) + 1};& 
		 \\ \text{(ii)}  && \left( 1 - 1 / \log n \right) \cdot \frac{1-Y^v_{k-1}}{k(n-k) + 1} & \leq \expect[\capped{X}^v_k \mid \A^v_{k-1}] 
		 \leq \frac{1}{k(n-k) + 1}.
	\end{flalign*}
\end{lemma}
\begin{proof}
	For every $k \in [n-1]$, we define the function $\w^v_k: S^v_k \rightarrow [0,1]$ as follows
	\[
	\w^v_k(e) = 
	\begin{cases}
		\lambda(e) - Y^v_{k-1}, & \quad \lambda(e) \geq Y^v_{k-1}\\
		\lambda(e) - Y^v_{k-1} + 1, & \quad \lambda(e) < Y^v_{k-1}.
	\end{cases}
	\]
	Notice that for any two edges $e,f \in S^v_k$ such that $\lambda(f) < Y^v_{k-1} \leq \lambda(e)$
	we have $\w^v_k(e) < \w^v_k(f)$. Furthermore, it follows from \cref{lem:foremostTree} that
	$S_k^v$ contains at least one edge $e$ with $\lambda(e) \geq Y_{k-1}^v$.
	These two observations imply that $e^v_k$ is exactly
	the edge on which the minimum of $w^v_k$ is attained, that is,
	\begin{equation}\label{eq:ek_is_argmin}
		e^v_k = \arg\min \{ \w^v_k(e) \mid e \in S^v_k \},
	\end{equation}
	and therefore, for every $k \in [n-1]$,
	\begin{equation}\label{eq:Xk}
		X^v_k = \min\{ w^v_k(e) \mid e \in S^v_k \}.
	\end{equation}
	
	Observe that upon exposure of edge~$e_k^v$ at step~$k$ of the algorithm, we reveal some information about the time labels of the other edges in $S_k^v$.
	More precisely, we learn that these time labels are not contained in the interval $[Y_{k-1}^v, Y_k^v]$.
	Thus, if, for every $k \in [n-1]$ and $e \in S_k^v$, we inductively define the admissible range of $\lambda(e)$, as $I^v_k(e) := [0,1]$ for $k=1$, and for every $2 \leq k \leq n-1$
	\[
	I^v_k(e) :=
	\begin{cases}
		I^v_{k-1}(e) \setminus [Y^v_{k-2}, Y^v_{k-1}], & \quad e \in S^v_k \cap S^v_{k-1}\\
		[0,1], & \quad e \in S^v_k \setminus S^v_{k-1}
	\end{cases}
	\]
	then $\lambda(e)$ conditioned on $\A_{k-1}^v$ is uniformly distributed on $I_k^v(e)$.
	Let $\ell = \ell(e)$ be the unique index with $e \in S_\ell^v \setminus S_{\ell-1}^v$.
	Then, we have $I_k^v(e) = [0,1] \setminus [Y_{\ell-1}^v , Y_{k-1}^v]$.
	Thus, by definition, if $\lambda(e) \geq Y_{k-1}^v$, then $w_k^v(e)$ is uniformly distributed on $[0, 1-Y_{k-1}^v]$,
	and if $\lambda(e) < Y_{k-1}^v$, then $w_k^v(e)$ is uniformly distributed on $[1-Y_{k-1}^v, 1+Y_{\ell-1}^v-Y_{k-1}^v]$.
	It follows that $w_k^v(e)$ is uniformly distributed on its admissible range
	\[
		J_k^v(e) := \mathopen[ 0, 1 + Y_{\ell-1}^v - Y_{k-1}^v \mathclose]
	\]
	and clearly
	\begin{equation}\label{eq:inclusion}
		\mathopen[ 0, 1 - Y_{k-1}^v \mathclose] \subseteq J_k^v(e) \subseteq [0, 1] \,.
	\end{equation}
	
	Note that $X^v_k$ is a minimum of $k(n-k)$ independent random variables $w^v_k(e), e \in S^v_k$,
	where for every edge $e \in S^v_k$ the value $w^v_k(e)$ is distributed uniformly on its own 
	admissible range $J^v_k(e)$.
	Let $X_k'$ be the minimum of $k(n-k)$ independent random variables distributed uniformly in 
	$\left[0, 1-Y^v_{k-1} \right]$, and $X_k''$ be the minimum of $k(n-k)$ independent
	random variables distributed uniformly in $[0, 1]$. 
	Then, the inclusion (\ref{eq:inclusion}) implies 
        \begin{equation}\label{eq:WaitTimePrimed}
		\prob[X_k' \geq t] \leq \prob[X^v_k \geq t \mid A^v_{k-1}] \leq \prob[X_k'' \geq t],
        \end{equation}
	where $A^v_{k-1}$ is the event that specific 
	edges $e^v_1, e^v_2, \ldots, e^v_{k-1}$ with time labels $Y^v_1, Y^v_2, \ldots, Y^v_{k-1}$ are revealed in the 
	first $k-1$ steps of the process. 
        In particular, $\prob(X^v_k\geq c_k)\leq (1-c_k)^{k(n-k)}$.
        Thus, by the union bound and \cref{lm:XCapAAS_sum},
        the probability that for some $k$ we have $X^v_k\ne\capped{X}^v_k$ 
        is at most $4/\log n$. This proves the first part of the lemma.
        
    Now, by integrating inequality (\ref{eq:WaitTimePrimed}) and noting that
	the expected value of the minimum of $m$ independent variables distributed uniformly on $[0,a]$ is equal to $\frac{a}{m+1}$,
	we obtain
	
	\begin{equation}
		\frac{1-Y^v_{k-1}}{k(n-k) + 1} 
		= 
		\expect[ X_k'] 
		\leq 
		\expect[ X^v_k \mid \A_{k-1}] 
		\leq 
		\expect[ X_k'']
		= 
		\frac{1}{k(n-k) + 1}.
	\end{equation}

	To prove item (ii) of the lemma, we first note that by definition $\capped{X}^v_k \leq X^v_k$, and hence
	\[
		\expect[ \capped{X}^v_k \mid \A^v_{k-1} ] 
		\leq 
		\expect[ X^v_k \mid \A^v_{k-1} ]
		\leq
		\frac{1}{k(n-k) + 1}.
	\]
	
	\noindent
	Therefore, it remains to show the lower bound on $\expect[ \capped{X}^v_k \mid \A^v_{k-1} ]$. For convenience, let us write $M_k$ for $k (n-k)$.
	Then, we have
	
	\begin{equation*}
	\begin{split}
		\expect[ \capped{X}^v_k \mid \A^v_{k-1} ] 
		&=
		\int_{0}^{\infty} \prob[\capped{X}^v_k \geq t \mid A^v_{k-1}] \dd t
		=
		\int_{0}^{c_k} \prob[\capped{X}^v_k \geq t \mid A^v_{k-1}] \dd t
		\\ &= 
		\int_{0}^{c_k} \prob[X^v_k \geq t \mid A^v_{k-1}] \dd t
		\geq
		\int_{0}^{c_k} \prob[X_k' \geq t] \dd t 
		\\ &=
		\int_{0}^{c_k} \left( 1 - \frac{t}{1 - Y^v_{k-1}} \right)^{M_k} \dd t
		\\ &=
		\frac{(1 - Y^v_{k-1}) + (c_k - (1 - Y^v_{k-1})) \left( 1 - \frac{c_k}{1 - Y^v_{k-1}} \right)^{M_k}}{M_k + 1}
		\\ & \geq
		\frac{(1 - Y^v_{k-1}) - (1 - Y^v_{k-1}) ( 1 - c_k )^{M_k}}{M_k + 1}
		\\ & =
		\frac{(1 - Y^v_{k-1}) \left( 1 - ( 1 - c_k )^{M_k} \right)}{M_k + 1},
	\end{split}
	\end{equation*}
    and the desired bound follows from the fact that
   $(1 - c_k)^{k(n-k)}  \leq  \frac{1}{\log n}$ (see (\ref{eq:prob_capped_uncapped_sum})).
\end{proof}

Recall that, by \cref{cor:2hopTempConn}, $\F_{n,\: (\log n)/\sqrt{n}}$ is temporally connected 
with probability at least $1-n^{-\frac{\log{n}}{4} + 2}$,
in which case the bound $Y^v_k \leq Y^v_{n-1} \leq \frac{\log{n}}{\sqrt{n}}$ holds for every $v$ and $k \in [n-1]$.
Thus, we derive the following

\begin{corollary}\label{cor:ExpOfcappedXk}
	With probability at least $1-n^{-\frac{\log{n}}{4} + 2}$, for every vertex $v$ and every $k \in [n-1]$, we have
	\begin{flalign*}
		\text{(i)} && 	\left( 1 - \frac{\log{n}}{\sqrt{n}} \right) \cdot \frac{1}{k(n-k) + 1}  &\leq  \expect[X^v_k \mid \A^v_{k-1}]  \leq  \frac{1}{k(n-k) + 1};& 
		\\ \text{(ii)}  && \left( 1 - \frac{2}{\log n} \right) \cdot \frac{1}{k(n-k) + 1}  &\leq  \expect[\capped{X}^v_k \mid \A^v_{k-1}]  \leq  \frac{1}{k(n-k) + 1}.
	\end{flalign*}
\end{corollary}

Next, we will bound the deviation of the truncated time of the moment when the foremost tree acquires $k$ edges from the expected value of accumulated truncated waiting times between the consecutive edges in the sequence of the first $k$ edges of the tree.
For this we require the following standard inequality by Azuma.
\begin{theorem}[Azuma's inequality~\cite{Azuma1967}]\label{the:Azuma}
	Let $Z_0, Z_1, \ldots, Z_n$ be a martingale with respect to a filtration 
	$\{\emptyset, \Omega\} = \A_0 \subset \A_1 \subset \cdots \subset \A_{n}$. 
	Let also $c_1, c_2, \ldots, c_n$ be non-negative numbers such that
	$
		\sum_{i=1}^{n} \prob \left[ \abs{Z_i - Z_{i-1}} \geq c_i \right] = 0.
	$
	Then
	\[
	\prob\bigl[ \abs{Z_n - Z_{0}} \geq \mu \bigr] \leq 2 \exp\left(\frac{-\mu^2}{2\sum_{i=1}^n c_i^2}\right).
	\]
\end{theorem}

\begin{lemma}
	\label{lm:YCapConcentrate}
	For a fixed vertex $v$ and any sufficiently large $n$, with probability at least $1 - n^{-\sqrt{\log n}-1}$ the inequality
	\[
		\left| \capped{Y}^v_k - \sum_{i=1}^k \expect \left[ \capped{X}^v_i \mid \A^v_{i-1} \right] \right| < \frac{(\log n)^{0.8}}{n}
	\] 
	holds for all $k \in [n-1]$.
\end{lemma}

\begin{proof}
	Let us fix $k \in [n-1]$
	and define a martingale $Z^v_0, Z^v_1, \ldots, Z^v_{k}$ with $Z^v_0 := 0$ and 
	\[
		Z^v_s := \capped{Y}^v_s - \sum_{i=1}^s \expect [ \capped{X}^v_i \mid \A^v_{i-1} ]
			   = \sum_{i=1}^s \capped{X}^v_i - \sum_{i=1}^s \expect[\capped{X}^v_i \mid \A^v_{i-1}],
	\]
	for $s \in [k-1]$. 
	
	Since 
	$0 \leq \capped{X}^v_i \leq {} c_i$, we have
	$0 \leq \expect[\capped{X}^v_i \mid \A^v_{i-1}] \leq {} c_i$,
	and therefore
	\[
	\prob\left[\abs*{Z^v_{i}-Z^v_{i-1}}>c_i \right]
	=
	\prob\left[ \abs*{\capped{X}^v_i - E[\capped{X}^v_i\mid\A^v_{i-1}]} >c_i\right]
	= 0
	\]
	holds for every $i \in [k]$.
	Furthermore, by \cref{lm:XCapAAS_sum} \ref{ck_bound}, for all sufficiently large $n$,
	\[
		\sum_{i=1}^{k} c_i^2 
		\leq 
		\sum_{i=1}^{n-1} c_i^2  
		\leq
		\frac{64 (\log \log{n})^2}{n^2},
	\]
	and hence applying Azuma's inequality (\cref{the:Azuma}), we obtain that for sufficiently large $n$
	\begin{equation*}
		\begin{split}
			\prob \left[ |Z^v_k|\geq\frac{(\log n)^{0.8}}{n} \right]
			&\leq
			2\exp\left(
			\frac{-(\log n)^{1.6}}{n^2} \cdot 
			\frac{1}{2\sum_{i=1}^{k}c_k^2}
			\right)
			\\ &\leq
			2\exp\left(
			\frac{-(\log n)^{1.6}}{n^2} \cdot 
			\frac{n^2}{128 (\log\log n)^2}
			\right)
			\\ &=
			2\exp\left(
			\frac{-(\log n)^{1.6}}{128 (\log\log n)^2}
			\right)
			\\ &\leq
			\exp\left(-(\log n)^{1.55}\right)
			\\ &\leq 
			n^{-\sqrt{\log n}-2}.
		\end{split}
	\end{equation*}
	
	\noindent
	The latter inequality together with the union bound over all $k \in [n-1]$ imply the desired result.
\end{proof}

The following technical inequality will be useful in the rest of the section.
\begin{lemma}
\label{cl: favsum}
	For all integers $n \geq 1$ and $0 \leq k \leq n-1$, we have
	\begin{align*}
		\frac{\log (k+1) + \log n - \log(n-k)}{n} - \frac{3}{n}
		&\leq \sum_{i=1}^k \frac{1}{i(n-i) + 1} \leq \sum_{i=1}^k \frac{1}{i(n-i)}\\
		&\leq \frac{\log (k+1) + \log n - \log(n-k)}{n} + \frac{3}{n}
	\end{align*}
\end{lemma}
\begin{proof}
	We have
	\begin{align*}
		\sum_{i=1}^k \frac{1}{i(n-i) + 1} &\leq \sum_{i=1}^k \frac{1}{i(n-i)}
		= \frac{1}{n} \sum_{i=1}^k \left( \frac{1}{i} + \frac{1}{n-i} \right)
		\\ &\leq \frac{\log (k+1) + \log n - \log(n-k)}{n} + \frac{3}{n}
		\intertext{as well as}
		\sum_{i=1}^k \frac{1}{i(n-i) + 1} &\geq  \sum_{i=1}^k \frac{1}{i(n+1-i)}
		= \frac{1}{n+1}  \sum_{i=1}^k \left( \frac{1}{i} + \frac{1}{n+1-i} \right) 
		\\ &\geq \frac{\log (k+1) + \log n - \log(n-k)}{n+1}  - \frac{1}{n+1}
		\\ &\geq \frac{\log (k+1) + \log n - \log(n-k)}{n} - \frac{2 \log n}{n^2} - \frac{1}{n} 
		\\ &\geq \frac{\log (k+1) + \log n - \log(n-k)}{n} - \frac{3}{n} \,.
	\end{align*}
\end{proof}

Now we are ready to prove the main results of this section.

\begin{theorem}\label{th:YhatConc}
	With probability at least $1 - 2n^{-\sqrt{\log n}}$, for every vertex $v$, any large enough $n$, and $k \in [n-1]$, we have
	\[
			\left| \capped{Y}^v_k - \sum_{i=1}^k \frac{1}{i(n-i) + 1} \right| < \frac{2(\log n)^{0.8}}{n}.
	\]
\end{theorem}
\begin{proof}
	By \cref{cor:ExpOfcappedXk}~(ii), with probability at least $1-n^{-\frac{\log{n}}{4} + 2}$,
	for every vertex $v$ and every $k \in [n-1]$, we have
	\begin{equation}\label{eq:diffYhatExpHarmonic}
		- \sum_{i=1}^k \expect[\capped{X}^v_i \mid \A^v_{i-1}] - \frac{2}{\log n} \cdot \sum_{i=1}^k \frac{1}{i(n-i) + 1} 
		\leq 
		- \sum_{i=1}^k \frac{1}{i(n-i) + 1}
		\leq 
		-\sum_{i=1}^k \expect[\capped{X}^v_i \mid \A^v_{i-1}] .
	\end{equation}
	
	Similarly, by the union bound and \cref{lm:YCapConcentrate}, with probability at least $1 - n^{-\sqrt{\log n}}$
	for every vertex $v$, any sufficiently large $n$, and every $k \in [n-1]$, we have
	\begin{equation}\label{eq:Yhat}
		\sum_{i=1}^k \expect[\capped{X}^v_i \mid \A^v_{i-1}] - \frac{(\log n)^{0.8}}{n} 
		\leq 
		\capped{Y}^v_k
		\leq 
		\sum_{i=1}^k \expect[\capped{X}^v_i \mid \A^v_{i-1}] + \frac{(\log n)^{0.8}}{n} .
	\end{equation}

	Hence, summing up (\ref{eq:diffYhatExpHarmonic}) and (\ref{eq:Yhat}), we conclude that with probability at least $1 - 2n^{-\sqrt{\log n}}$ for every vertex $v$ and every $k \in [n-1]$
	\begin{equation*}\label{eq:YhatHarmonic}
		- \frac{(\log n)^{0.8}}{n} - \frac{2}{\log n} \cdot \sum_{i=1}^k \frac{1}{i(n-i) + 1}
		\leq 
		\capped{Y}^v_k - \sum_{i=1}^k \frac{1}{i(n-i) + 1}
		\leq 
		\frac{(\log n)^{0.8}}{n},
	\end{equation*}
	which implies the result after noticing that by \cref{cl: favsum} and for large enough $n$
	\begin{align*}
		 \frac{2}{\log n} \cdot \sum_{i=1}^k \frac{1}{i(n-i) + 1} 
		\leq
		\frac{2}{\log n} \cdot
		\left(  \frac{2 \log n}{n} + \frac{3}{n}
		\right)		
		 \leq \frac{10}{n}
		\leq \frac{(\log n)^{0.8}}{n}.
	\end{align*}
\end{proof}

\begin{corollary}
	\label{lm:LastYCappedConcentrate}
	With probability at least $ 1 - 2n^{-\sqrt{\log n}}$, for every vertex $v$ and any large enough $n$, we have
	\[
		\left| \capped{Y}^v_{n-1} - \frac{2\log n}{n} \right|
		<
		\frac{3(\log n)^{0.8}}{n}.
	\]
\end{corollary}
\begin{proof}
	By \cref{cl: favsum} we have
	
	\[
		\left| \sum_{i=1}^{n-1}\frac{1}{i(n-i)+1} - \frac{2\log n}{n} \right| \leq \frac {3}{n}.
	\]
	Hence, from \cref{th:YhatConc}
	we conclude that with probability at least $1 - 2n^{-\sqrt{\log n}}$ the inequality
	\begin{align*}
		\left| \capped{Y}^v_{n-1} - \frac{2\log n}{n} \right|
		&\leq
		\left| \capped{Y}^v_{n-1}-\sum_{i=1}^{n-1}\frac{1}{i(n-i)+1} \right|
		+
		\left| \sum_{i=1}^{n-1}\frac{1}{i(n-i)+1}-\frac{2\log n}{n} \right|
		\\ &\leq
		\frac{2(\log n)^{0.8}}{n} + \frac{3}{n}
		<
		\frac{3(\log n)^{0.8}}{n}
	\end{align*}
	 holds for every vertex $v$ and any large enough $n$.
\end{proof}

\cref{th:YhatConc} and \cref{lem:ExpOfXk} (introductory claim) imply the following analog of \cref{th:YhatConc}
for the non-truncated random variables.
\begin{theorem}\label{th:YConc}
	For a fixed vertex $v$ and any large enough $n$,
	with probability at least $1 - 5 / \log n$, 
	for every $k \in [n-1]$, we have
	\[
	\left| Y^v_k - \sum_{i=1}^k \frac{1}{i(n-i) + 1} \right| < \frac{2(\log n)^{0.8}}{n}.
	\]
\end{theorem}
\noindent
The latter statement together with \cref{cl: favsum} immediately implies the following

\begin{corollary}\label{th:YConcEasy}
	For a fixed vertex~$v$ and any large enough $n$, with probability at least $1 - 5/\log n$, for every $0 \leq k \leq n-1$, we have
	\[
	\abs*{Y_k^v - \frac{\log n + \log (k+1) - \log (n-k)}{n}} < \frac{2 (\log n)^{0.8} + 3}{n}.
	\]
\end{corollary}

\noindent
Finally, we use \cref{th:YConcEasy} to prove the following lemma that will be utilized repeatedly throughout the paper.

\begin{lemma}\label{th:reachability}
	For every function $z = z(n)$ with $0 \leq z(n) \leq 1$, and any constant $y > 0$ there is $n_0 := n_0(y)$ such that for all $n \geq n_0$
	a fixed vertex $v$ in $\G \sim \F_{n,p}$ 
	can reach (resp.\ be reached by)
	\begin{enumerate}[(i)]
		\item \emph{at least} $\ceil*{\frac{n^z}{(\log n)^y}}$ 
		vertices with probability at least $1 - \frac{5}{\log{n}}$, if $p \geq z\frac{\log n}{n} + \frac{3(\log n)^{0.8}}{n}$;
		\label{th:reachability-1}
		
		\item \emph{fewer than} $\ceil*{\frac{n^z}{(\log n)^y}}$ 
		vertices with probability at least $1 - \frac{5}{\log{n}}$, if $p \leq z\frac{\log n}{n} - \frac{3(\log n)^{0.8}}{n}$;
		\label{th:reachability-2}
	
		\item \emph{at least} $n - \left\lfloor\frac{n^z}{(\log n)^y}\right\rfloor$ 
		vertices with probability at least $1 - \frac{5}{\log{n}}$, if $p \geq (2-z)\frac{\log n}{n} + \frac{3(\log n)^{0.8}}{n}$;
		\label{th:reachability-3}
		
		\item \emph{fewer than} $n - \left\lfloor\frac{n^z}{(\log n)^y}\right\rfloor$ 
		vertices with probability at least $1 - \frac{5}{\log{n}}$, if $p \leq (2-z)\frac{\log n}{n} - \frac{3(\log n)^{0.8}}{n}$.
		\label{th:reachability-4}
	\end{enumerate}
\end{lemma}
\begin{proof}
	Define $d := \ceil*{\frac{n^z}{(\log n)^y}}$.
	By choosing $n_0$ large enough, we may assume $1 \leq d \leq n-1$.
	To prove the ``can reach'' part of (i) and (ii), it suffices to show that for
	a fixed vertex $v$, with probability at least $1 - 5 / \log{n}$, we have
	\[
	\left| Y^v_{d-1} - z\frac{\log n}{n} \right|
	<
	\frac{3(\log n)^{0.8}}{n} \,.
	\]
	Observe that
	\begin{align*}
		L \coloneqq{}& \log n + \log d - \log(n-d + 1)
		\\ ={}& \log n + \log\left(\frac{n^z}{(\log n)^y} + 1 \right) - \log\left(n-\frac{n^z}{(\log n)^y}\right) + \bigO(1)
		\\ ={}& z \log n +  \log\left(\frac{1}{(\log n)^y} + \frac{1}{n^z} \right) - \log\left(1-\frac{n^{z-1}}{(\log n)^y}\right) + \bigO(1)
		\\ ={}& z \log n + \bigO(\log \log n).
	\end{align*}
	Then we conclude by \cref{th:YConcEasy} that with probability at least $1 - 5 / \log{n}$
	\begin{align*}
		\abs*{ Y^v_{d-1} - z\frac{\log n}{n} }
		&\leq
		\abs*{ Y^v_{d-1}- \frac{L}{n} }
		+
		\abs*{ \frac{L}{n} - z\frac{\log n}{n} }
		\\ &\leq
		\frac{2(\log n)^{0.8} + 3}{n} + \frac{\bigO(\log \log n)}{n}
		<
		\frac{3(\log n)^{0.8}}{n}
	\end{align*}
	holds, given that $n_0$ (and thus~$n$) is at least a sufficient constant.
	
	The ``can reach'' statement of (iii) and (iv) is proven analogously by showing that for $d' := \floor*{\frac{n^z}{(\log n)^y}}$
	we have
	\begin{align*}
		L' \coloneqq{}& \log n + \log(n-d') - \log(d'+1)
		\\ ={}& (2-z) \log n + \bigO(\log\log n).
	\end{align*}
	and thus with probability at least $1 - 5 / \log{n}$
	\begin{align*}
		\abs*{ Y^v_{n-d'-1} - (2-z)\frac{\log n}{n} }
		\leq
		\abs*{ Y^v_{n-d'-1}- \frac{L'}{n} }
		+
		\abs*{ \frac{L'}{n} - (2-z)\frac{\log n}{n} }
		<
		\frac{3(\log n)^{0.8}}{n} \,.
	\end{align*}
	Finally, the ``can be reached by'' part of the theorem holds by symmetry under time reversal.
\end{proof}

\subsection{Foremost tree height}\label{sec:foremostTreeHeight}

In this section we study the height of the foremost trees $T_{n-1}^v$
in a random temporal complete graph $(G, \lambda) \sim \F_{n,1}$.
We prove that a.a.s.\ each foremost tree $T_{n-1}^v$ has at most logarithmic height.

We say that an endpoint $u$ of $e_k^v$ is the \emph{attachment vertex of $e_k^v$}, if $u$ belongs to $T_{k-1}^v$.
Recall that $e^v_1, e^v_2, \ldots, e^v_{k-1}$ denote the edges that are revealed in the first $k-1$ steps of the process and that form $T_{k-1}^v$; and $Y^v_1, Y^v_2, \ldots, Y^v_{k-1}$
denote the time labels of these edges, respectively.
In the next lemma we show that the attachment vertex of $e_k^v$ is distributed almost uniformly in $V(T_{k-1}^v)$. Denote by $\mathcal{T}^v_{k-1}$ the family of all $k$-vertex trees rooted at $v$ and with vertices from $[n]$.
\begin{lemma}\label{lm:foremost-heights-step-dist}
	There exists a function $\eps = \eps(n) \in O\left(\frac{\log n}{\sqrt{n}}\right)$ such that for all 
	$n \geq 2$, $ k \in [n-1]$, and any $T \in \mathcal{T}^v_{k-1}$
	\[
                \max_{u \in V\left(T\right)} \abs*{ \prob \left[ u \in e_k^v \mid T_{k-1}^v = T \right] - \frac{1}{k}} \leq \frac{\eps}{k} \,.
	\]
\end{lemma}
\begin{proof}
	Recall that the event $Y_{n-1}^v \leq \frac{\log n}{\sqrt{n}}$
        happens with probability at least $1-n^{-\frac{\log{n}}{4} + 2} =: 1 - \eps'$ by \cref{cor:2hopTempConn}.
	Using the notation $\w^v_k(e)$ and $J_k^v(e)$ introduced in the proof of \cref{lem:ExpOfXk},
	recall that for each $k \in [n-1]$, we have
	\[
	e^v_k = \arg\min \{ \w^v_k(e) \mid e \in S^v_k \},
	\]
	where the variable $\w^v_k(e)$ is uniformly distributed on $J_k^v(e)$.
        As long as $Y_{k-1}^v\leq\frac{\log n}{\sqrt{n}}$, by (\ref{eq:inclusion}), the interval $J_k^v(e)$
        has length between $1-\frac{\log n}{\sqrt{n}}$ and $1$.
	
	Since, for each edge~$e \in S_k^v$, the probability of the event $e = e_k^v$ 
	is inversely proportional to the length of $J_k^v(e)$,
        conditional on $Y_{k-1}^v\leq\frac{\log n}{\sqrt{n}}$
	this probability can deviate from $1/\abs{S_k^v}$ at most by a factor of $1-\frac{\log n}{\sqrt{n}} =: 1 -\eps''$.
	As each vertex of $T_{k-1}^v$ is incident with the same number of edges in $S_k^v$,
	the probability of some vertex $u$ of the tree being incident with $e_k^v$ is thus within a factor of $1-\eps''$ of $\frac{1}{k}$.
      Consequently, we have
	\[
	\frac{1-\eps''}{k} \leq \prob \left[ u \in e_k^v \mid T_{k-1}^v=T, Y_{k-1}^v \leq \frac{\log n}{\sqrt{n}} \right] \leq \frac{1}{k (1-\eps'')}.
	\]
	Now, since the event $Y_{k-1}^v \leq \frac{\log n}{\sqrt{n}}$ is implied by
	$Y_{n-1}^v \leq \frac{\log n}{\sqrt{n}}$, it happens with probability at least $1-\varepsilon'$,
	and therefore
	\[
	(1-\eps')\frac{1-\eps''}{k} \leq \prob[u \in e_k^v \mid T_{k-1}^v=T] \leq \frac{1-\eps'}{k (1-\eps'')} + \eps' . 
	\]
	When setting $\eps := n \eps' + 4 \eps'' \in O\left(\frac{\log n}{\sqrt{n}}\right)$,
	this yields
	\[
	\abs*{ \prob[u \in e_k^v \mid T_{k-1}^v=T] - \frac{1}{k}} \leq \frac{\eps}{k} \,,
	\]
	where we used that
	\[
		(1-\eps') \frac{1-\eps''}{k} 
		>
		\frac{1 - \eps' - \eps''}{k}
		>
		\frac{1-\eps}{k}
	\]
	and, since $\eps'' < \frac{3}{4}$,
	\[
		\frac{1-\eps'}{k (1-\eps'')} + \eps'
		<
		\frac{1}{k (1-\eps'')} + \eps'
		<
		\frac{1 + 4 \eps''}{k} + \eps'
		<
		\frac{1 + \eps}{k}	\,.
	\]
\end{proof}

Let now $N(h,k,n)$ denote
the number
of vertices of height $h$
in the $k$-vertex partial foremost tree $T_{k-1}^v$
in $(G, \lambda) \sim \F_{n,1}$,
where the height of a vertex is the distance from the vertex to the root $v$.

\begin{lemma}
	\label{lm:foremost-height-step}
	There exists a positive function $\eps = \eps(n) \in O\left(\frac{\log n}{\sqrt{n}}\right)$
	such that for all $n \geq 2$, $k \in [n-1]$, and $h \in [k]$,
	\[
	\frac{1-\eps}{k}\expect[N(h-1,k,n)]
	\leq \expect[N(h,k+1,n)]  - \expect[N(h,k,n)] \leq
	\frac{1+\eps}{k}\expect[N(h-1,k,n)] \,.
	\]
\end{lemma}
\begin{proof}
	Let $\eps$ be as in \cref{lm:foremost-heights-step-dist}.
	The result follows from the observation that
	$N(h, k+1, n) - N(h, k, n)$ is either $1$ or $0$,
	depending on whether the edge $e_k^v$ was attached to a vertex of height $h-1$ or not,
	which, by \cref{lm:foremost-heights-step-dist}, happens with probability between $(1 - \eps)N(h-1, k, n)/k$ and $(1+\eps)N(h-1, k, n)/k$.
\end{proof}

\begin{lemma}
	\label{lm:foremost-height-upper-bounds}
	For all $n$ sufficiently large,
	$k \in [n-1]$, and $0 \leq h \leq k$,
	we have the inequality
	$$
	\expect[N(h,k,n)]
	\leq
	\frac{(4\log k)^h}{h!}.
	$$
\end{lemma}
\begin{proof}
	We prove the lemma by double induction on $h,k$.
	Clearly, the statement holds whenever $k=1$ or $h=0$.
	Hence, assuming $h \geq 1$, we derive
	\begin{align*}
		\expect[N(h,k+1,n)]
		&\leq
		\expect[N(h,k,n)] + \frac{2}{k}\expect[N(h-1,k,n)]
		\quad\text{by \cref{lm:foremost-height-step}}
		\\ &\leq
		\frac{(4\log k)^h}{h!}
		+
		\frac{2(4\log k)^{h-1}}{k(h-1)!} \quad\text{by induction hypothesis}
		\\ &=
		\frac{(4\log k)^h}{h!}
		\left(1+\frac{2h}{k(4\log k)}\right)
		\\ &\leq
		\frac{(4\log k)^h}{h!}
		\left(1+\frac{2}{k(4\log k)}\right)^h
		\\ &=
		\frac{(4 \log k + 2/k)^h}{h!}
		\\ &\leq
		\frac{(4\log (k+1))^h}{h!}
	\end{align*}
	where the last inequality holds because
	\[
	4\log k+\frac{2}{k}
	=
	4\left(\log k+\frac{1}{2k}\right)
	\leq
	4\left(\log k+\frac{1}{k+1}\right)
	\leq
	4\log(k+1) \,.
	\]
\end{proof}

\noindent
On a side note, we believe the following more precise asymptotic estimate of $\expect[N(h,k,n)]$ to be true.

\begin{conjecture}
	\label{conj:foremost-heights-ideal}
	\[
	\lim_{k \to \infty} \sup_{\substack{h,n \\ h\leq k \leq n}} \left| \expect[N(h,k,n)]-\frac{(\log k)^h}{h!} \right| = 0 \,.
	\]
\end{conjecture}

We are now ready to state and prove the main result of the section.

\begin{theorem}
	\label{the:foremost-heights-aas}
	Let $\G = (G, \lambda) \sim \F_{n,1}$.
	Then \aas for all vertices $v$ in $\G$ and all $0 \leq \tau \leq \tau' \leq 1$ simultaneously,
	the foremost tree for $v$ in $\G_{[\tau,\tau']}$ has height at most $14 \log n$.
\end{theorem}
\begin{proof}
	To prove the theorem, we will show that the foremost trees constructed by \cref{alg:foremostTree}
	have the desired height.
	First, without loss of generality we can assume that 
	$\tau, \tau' \in \Lambda := \{ \lambda(e) \,|\, e \in E(G) \}$.
	Furthermore, it follows from \cref{alg:foremostTree} that the foremost tree for $v$ in $\G_{[\tau,\tau']}$
	is a subtree of the foremost tree for $v$ in $\G_{[\tau,1]}$. 
	Hence, it suffices to show the claim when $\tau' = 1$.
	Let $\ell_\tau(v)$ denote the number of vertices of height exactly $\ceil{14 \log n}$ in $T_{v,\tau}$,
	which denotes the foremost tree for~$v$ in $\G_{[\tau, 1]}$.
	To prove the theorem, it suffices to show that
	$s := \max \{ \ell_\tau(v) \,|\, v \in V, \tau \in \Lambda \}$ is zero a.a.s.
	
	Since $\G_{[\tau,1]}$ is distributed according to $\F_{1-\tau, n}$ (up to rescaling),
	the number of vertices of height $h$ in $T_{v,\tau}$ is distributed as 
	the number of vertices of height $h$ in the foremost tree for $v$ in $\G_{[0,1-\tau]}$,
	which is upper-bounded by the number of vertices of height $h$ in the foremost tree $T_{n-1}^v$ of $\G$.
	Therefore we obtain from \cref{lm:foremost-height-upper-bounds} that
	
	\begin{align*}
		\expect[s] 
		&\leq 
		\expect\left[\sum_{v \in V}\sum_{\tau \in \Lambda} \ell_\tau(v)\right]
		\leq
		n^3 \cdot \expect\left[N\left( \ceil{14 \log n},n,n \right)\right]
		\leq
		n^3 \frac{(4 \log n)^{\ceil{ 14 \log n }}}{\ceil{ 14 \log n } !}
		\\ &\leq 
		n^3 \left(\frac{4e \log n}{\ceil{ 14 \log n } }\right)^{\ceil{ 14 \log n }} 
		\leq 
		n^3 \left(\frac{4e}{14}\right)^{14 \log n} 
		\quad\text{since } \frac{4e}{14} < 1
		\\ &=
		n^{17} \left(\frac{2}{7}\right)^{14 \log n}	
		= n^{17 + 14 \log(2/7)}
		\approx n^{-0.54}
		\in o(1).
	\end{align*}
	Thus, by Markov's inequality, \aas $s = 0$ as required.
\end{proof}

A theorem similar to \cref{the:foremost-heights-aas} holds also for hindmost trees.
The following result is an implication of these theorems and 
the observation that a temporal graph $\G$ contains a foremost (resp.\ hindmost) 
$(v,u)$-path if and only if any foremost tree for $v$ (resp.\ hindmost tree for $u$) in $\G$ 
contains a $(v,u)$-path.

\begin{theorem}
	\label{cor:foremost-len-aas}
	Let $\G \sim \F_{n,p}$. Then \aas for every ordered pair of vertices $v, u$ and
	all $\tau_1, \tau_2 \in [0,p]$, $\tau_1 \leq \tau_2$, the temporal graph $\G_{[\tau_1, \tau_2]}$
	either contains a foremost (resp.\ hindmost) $(v,u)$-path with fewer than $14 \log n$
	edges or contains no temporal $(v,u)$-path.
\end{theorem}

\begin{remark}
        It can also be shown that
        the foremost tree has a vertex of height at least $\log n/4$ \aas,
        and hence the foremost tree has logarithmic height \aas
        Indeed, the expected number of vertices \emph{below} this height is sublinear,
        thus by Markov's inequality there are vertices above this height \aas
\end{remark}

\section{Sharp thresholds for temporal graph properties}
\label{sec:tempProp}

In this section, we apply the results obtained in \cref{sec:foremostGrowth} to establish
sharp thresholds for Point-to-point Reachability, First Temporal Source, Temporal Source, and
Temporal Connectivity properties.
We recall that our general strategy for obtaining a sharp threshold $p_0$ for a certain property in
the model $\F_{n,p}$ is to show that the property \emph{does not} hold in a random temporal complete graph $\F_{n,1}$ before time $p_0$, and holds after time $p_0$ a.a.s.
Therefore, in the proofs, if we do not specify explicitly the model a graph under consideration comes from, we assume that it comes from $\F_{n,1}$.

\subsection{Point-to-point Reachability}

Recall that Point-to-point Reachability is the property that for a fixed pair of vertices there exists a temporal path from the first vertex to the second one.
In this section we establish a sharp threshold for this property. 

\begin{theorem}
\label{th:pathUV}
The function $\frac{\log n}{n}$ is a sharp threshold for Point-to-point Reachability.
More specifically, for any sufficiently large $n$,
for two fixed distinct vertices $v$ and $u$ in $\G \sim \F_{n,p}$
\begin{enumerate}[(i)]
\item
there is \emph{no} temporal $(v,u)$-path in $\G$
with probability at least $1 - \frac{6}{\log n}$,
if $p\leq\frac{\log n}{n}-\eps$; and
\item
there is a temporal $(v,u)$-path in $\G$
with probability at least $1-\frac{6}{\log n}$,
if $p\geq\frac{\log n}{n}+\eps$,
\end{enumerate}
where $\eps := \frac{3(\log n)^{0.8}}{n} \in o\left(\frac{\log n}{n}\right)$.
\end{theorem}
\begin{proof}
Assume $n$ to be sufficiently large.
By \cref{th:reachability} \ref{th:reachability-2} (with $z=y=1$), if $p \leq \frac{\log n}{n} - \eps$, then 
with probability at least $1-\frac{5}{\log n}$
vertex~$v$ reaches fewer than~$\frac{n}{\log n}$ vertices.
Analogously, by \cref{th:reachability} \ref{th:reachability-3} (with $z=y=1$), if $p \geq \frac{\log n}{n} + \eps$, then
with probability at least $1-\frac{5}{\log n}$
vertex~$v$ reaches all but $\frac{n}{\log n}$ vertices.

By symmetry, the probability of $u$ to be among the 
first $\frac{n}{\log n}$~vertices to become reachable from $v$
is equal to $\frac{n}{n \log n}=\frac{1}{\log n}$
and the same holds for the last $\frac{n}{\log n}$~vertices.
Hence, by the union bound, with probability at least $1- \frac{6}{\log n}$,
vertex~$u$ is not reachable from~$v$ if $p \leq \frac{\log n}{n} - \eps$.
Analogously and with the same probability, vertex~$u$ is reachable from~$v$ if $p \geq \frac{\log n}{n} + \eps$.

\end{proof}

\subsection{Temporal Source and First Temporal Source}
\label{sec:TempSource}

In this section, we first establish a sharp threshold for Temporal Source, i.e., for the property that a fixed vertex is a temporal source. Then, we show that, quite surprisingly, the same function happens to be a sharp threshold for the property of having at least one temporal source in the graph, i.e., for First Temporal Source.

Before proceeding we make the simple but subsequently useful observation that the probability of a vertex $v$
of $\F_{n,p}$ being a temporal source is equal to the probability of $v$ being a temporal sink. This follows
from the fact that $v$ is a temporal source in $(G, \lambda)$ if and only if it is a temporal sink in $(G, \lambda')$, where $\lambda'(e) := 1 - \lambda(e)$ for every edge of $G$.
Therefore, the next two theorems are proved for temporal source, but they are stated and valid also for temporal sink.

\begin{theorem}
\label{th:avgSource}

The function $\frac{2\log n}{n}$ is a sharp threshold for Temporal Source (resp. Temporal Sink). 
More specifically, for any sufficiently large $n$,
a fixed vertex $v$ in $\G \sim \F_{n,p}$
\begin{enumerate}[(i)]
\item
is  \emph{not} a temporal source (resp. temporal sink)
with probability at least  $1 - \frac{5}{\log n}$,
if $p \leq\frac{2\log n}{n}-\eps$;
\item
is a temporal source (resp. temporal sink)
with probability at least $1-\frac{5}{\log n}$,
if $p\geq\frac{2\log n}{n}+\eps$,
\end{enumerate}
where $\eps := \frac{3(\log n)^{0.8}}{n} \in o\left(\frac{\log n}{n}\right)$.
\end{theorem}
\begin{proof}
Follows directly from \cref{th:reachability}~\ref{th:reachability-3} and \ref{th:reachability-4} when setting $z := 0$.
\end{proof}

\begin{theorem}
\label{th:firstSource}

The function $\frac{2\log n}{n}$ is a sharp threshold for First Temporal Source (resp. First Temporal Sink). 
More specifically, for any sufficiently large $n$,
a random temporal graph $\G \sim \F_{n,p}$
\begin{enumerate}[(i)]
\item
\emph{does not} contain a temporal source (resp. temporal sink)
with probability at least $1-2n^{-\sqrt{\log n}}$,
if $p\leq\frac{2\log n}{n}-\eps$;
\item
contains a temporal source (resp. temporal sink)
with probability at least $1-\frac{5}{\log n}$,
if $p\geq\frac{2\log n}{n}+\eps$,
\end{enumerate}
where $\eps := \frac{3(\log n)^{0.8}}{n} \in o\left(\frac{\log n}{n}\right)$.
\end{theorem}
\begin{proof}
We observe that the upper bound of $\frac{2\log n}{n}+\eps$ on the threshold follows from \cref{th:avgSource},
and we only need to show the lower bound.
Recall that by definition we always have $\capped{Y}_{n-1}^v\leq{}Y_{n-1}^v$.
Therefore, by \cref{lm:LastYCappedConcentrate}, 
with probability at least $1-2n^{-\sqrt{\log n}}$ the inequality
$Y_{n-1}^v\geq\capped{Y}_{n-1}^v >\frac{2\log n}{n}-\frac{3(\log n)^{0.8}}{n}$
holds for all vertices $v$. In other words, with probability at least $1-2n^{-\sqrt{\log n}}$
no vertex is a temporal source
until time $\frac{2\log n}{n}-\frac{3(\log n)^{0.8}}{n}$.
\end{proof}

\subsection{Temporal Connectivity}

In this section, we establish a sharp threshold for Temporal Connectivity.

\begin{theorem}\label{th:tempConnectivity}
	The function $\frac{3 \log n}{n}$ is a sharp threshold for Temporal Connectivity.
	More specifically, for any sufficiently large $n$, a random temporal graph in $\F_{n,p}$
	\begin{enumerate}[(i)]
		\item is \emph{not} temporally connected \aas, if $p \leq \frac{3\log n}{n} -\eps$,
		where $\eps := \frac{6 (\log n)^{0.8}}{n} \in o\left(\frac{\log n}{n}\right)$;
		
		\item is temporally connected \aas, if $p \geq \frac{3\log n}{n} + \eps$,
		where $\eps := \frac{3 (\log n)^{0.8}}{n} \in o\left(\frac{\log n}{n}\right)$.
	\end{enumerate}
\end{theorem}

We split the proof of \cref{th:tempConnectivity} into two parts. 
Note that a temporal graph is temporally connected if and only if each of its vertices is a temporal sink.
Our strategy is to show that on the one hand before time $\frac{3 \log n}{n}$ at least one vertex in
a random temporal complete graph $\G \sim \F_{n,1}$ is not a temporal sink \aas (\cref{lm:tconnLower}), 
and on the other hand after time $\frac{3 \log n}{n}$ all vertices in $\G$ are temporal sinks \aas (\cref{lem:connUpperBound}).

\begin{lemma}
\label{lm:tconnLower}
Let $p \leq \frac{3\log n}{n} -\eps$, where $\eps := \frac{6 (\log n)^{0.8}}{n}$,
and let $\G \sim \F_{n,1}$.
Then, \aas the temporal graph $\G_{[0, p]}$ contains at least one vertex which is not a temporal sink.
\end{lemma}
\begin{proof}
Let $q$ be the time at which the first vertex in $\G$ becomes a temporal sink,
i.e., 
\[
	q := \min\{q' \mid \G_{[0, q']} \text{ has a temporal sink}\}.
\]
By \cref{th:firstSource}, $q \geq \frac{2 \log n}{n} - \frac{3 (\log n)^{0.8}}{n} = \frac{2 \log n}{n} - \frac{\eps}{2}$ a.a.s.
Note that $\G_{[0, q]}$ has at most two temporal sinks,
since the addition of a single edge can not turn more than two vertices into temporal sinks.
Furthermore, in order for a vertex which is not a temporal sink at time $q$ to become a temporal sink it has to acquire
at least one edge incident to it after time $q$. Thus, to prove the lemma,
we will show that the underlying graph $H$ of $\G_{[q,p]}$ has at least 3 isolated vertices, and hence at least one of these vertices is not a temporal sink in $\G_{[0, p]}$.

Thereto, note that any pair of vertices forms an edge in $H$ with probability $\gamma := \frac{p - q}{1-q}$,
unless it is an edge in $\G_{[0, q]}$.
Thus, graph $H$ follows the same distribution as drawing a random graph from $G_{n, \gamma}$
and then deleting all edges contained in $\G_{[0, q]}$.

As $\gamma = \frac{p-q}{1-q}$ is maximized when $q$ is minimal, we have for sufficiently large $n$ that
\begin{align*}
	\gamma &= \frac{p-q}{1-q} 
	\leq \frac{3 \frac{\log n}{n} - \eps - 2 \frac{\log n}{n} + \eps/2}{1 - 2 \frac{\log n}{n} + \eps/2}
	< \frac{ \frac{\log n}{n} - \frac{\eps}{2}}{1 - 2\frac{\log n}{n}}
	\\ &= \frac{1}{n} \cdot \frac{\log n - 3 (\log n)^{0.8}}{1 - 2 \frac{\log n}{n}}
	= \frac{1}{n} \left( \log n -  \frac{3 (\log n)^{0.8} - 2 \frac{(\log n)^2}{n}}{1 - 2 \frac{\log n}{n}} \right) 
	\\ &< \frac{1}{n} \left( \log n - \frac{2 (\log n)^{0.8} }{1 - 2 \frac{\log n}{n}} \right)
	< \frac{\log n - (\log n)^{0.8}}{n} \,.
\end{align*}

It is known that $G_{n, \gamma}$ contains \aas more than two isolated vertices if $\lim_{n \to \infty} n (1-\gamma)^{n-1} = \infty$ \cite[Theorem~3.1(ii)]{BollobasRandomGraphs}.
In order to show the latter, we first evaluate
\begin{align*}
	\tau &:= \log n + n \cdot \log  \left( 1 - \frac{\log n - (\log n)^{0.8}}{n} \right)
	= (\log n)^{0.8} - \bigO\left(\frac{(\log n)^2}{n}\right) 
	\xrightarrow{n \to \infty} \infty
\end{align*}
by means of Maclaurin expansion.
Now, using this, we derive
\begin{align*}
	n (1-\gamma)^{n-1} &\geq n \left( 1 - \frac{\log n - (\log n)^{0.8}}{n} \right)^{n-1}
	\\ &= \exp\left( \log n + (n-1) \log \left( 1 - \frac{\log n - (\log n)^{0.8}}{n} \right) \right)
	\\ &\geq e^\tau
	 \xrightarrow{n \to \infty} \infty \,.
\end{align*}
Thus, $G_{n,\gamma}$, and therefore $H$, contains at least three isolated vertices \aas, as required.
\end{proof}

\begin{lemma}
	\label{lem:connUpperBound}
	Let $p \geq \frac{3\log n}{n} + \eps$, where $\eps := \frac{3 (\log n)^{0.8}}{n}$,
	and let $\G \sim \F_{n,1}$.
	Then, \aas every vertex in $\G_{[0, p]}$ is a temporal sink.
\end{lemma}
\begin{proof}
	Let $q := \frac{2 \log n}{n}+ \eps$
	and $r := \frac{n}{\log n} \log\log\log n$.
	It follows from \cref{th:avgSource}
	that the probability for an arbitrary vertex \emph{not} to be temporal sink in $\G_{[0, q]}$ is at most $\frac{5}{\log n}$.
	Hence, the expected number of vertices that are not temporal sinks in $\G_{[0, q]}$ is at most $\frac{5n}{\log n}$, and therefore, by Markov's inequality, the probability of $\G_{[0, q]}$ having more than $r$ vertices that are not temporal sinks is at most
	$\frac{5n}{\log n} \cdot \frac{\log n}{n \cdot \log\log\log n}
	= \frac{5}{\log\log\log n} \in o(1)$. We will show that \aas no more than $\frac{\log n}{n}$ extra time is required for these at most $r$
	remaining vertices to become temporal sinks.
	
	Let $S \subseteq V$ be the set of vertices that are temporal sinks in $\G_{[0, q]}$.
	Observe that if at some time $t > q$ a vertex $w \in V \setminus S$ acquires an edge that connects $w$ with a temporal sink
	in $S$, then $w$ becomes a temporal sink no later than time $t$ (notice, that $w$ can become a temporal sink
	before time $t$). Therefore, to prove the lemma, we will show that \aas by time $p = q + \frac{\log n}{n}$ every vertex in 
	$V \setminus S$ acquires such an edge.

	First, note that the underlying graph $G_{[0, q]}$ of $\G_{[0, q]}$ is distributed according to $G_{n,q}$.
	Hence, we may apply \cite[Corollary~3.13]{BollobasRandomGraphs} which states that the maximum degree $\Delta(G_{[0, q]})$ of $G_{[0, q]}$ \aas satisfies
	\begin{align*}
		\abs*{ \Delta(G_{[0, q]}) - qn - \sqrt{2q(1-q)n \cdot \log n} 
			+ \log\log n \sqrt{\frac{q(1-q)n}{8 \log n}}
			+ \log\left(2 \sqrt{\pi}\right) \sqrt{\frac{q(1-q)n}{2 \log n}}
			} \\ \leq  \frac{\log n}{2} \sqrt{\frac{q (1-q) n}{\log n}}.
	\end{align*}
	Thus, \aas we have
	\begin{align*}
		\Delta(G_{[0, q]})
		&\leq qn + \sqrt{2q(1-q)n \cdot \log n} + \frac{\log n}{2} \sqrt{\frac{q (1-q) n}{\log n}}
		\\ &< 3 \log n + \sqrt{6} \log n + \frac{\sqrt{3}}{2} \log n
		< 7 \log n
	\end{align*}
	where we used that $\eps < \frac{\log n}{n}$ for sufficiently large $n$.
	For every vertex $w \in V \setminus S$ define
	\[
		C_w := \{ vw \mid v \in S \text{ and } \lambda(vw) > q\}
	\]
	as the set of edges connecting $w$ to $S$ after time $q$.
	By the above, \aas for every $w$, we have that $\abs{C_w} \geq n - 1 - r - 7 \log(n) \geq n - 2r =: d$.
	Let $T_w$ be the \emph{waiting time} for the first of these edges to appear, i.e.,
	$T_w := \min \{ \lambda(vw) - q \mid vw \in C_w \}$.
	Note that the time labels of the edges in $C_w$ are independently and uniformly distributed on the interval $I = (q, 1]$.
	Thus, we have $\prob[T_w > x] \leq (1-x)^d$ for any $0 \leq x \leq 1-q$.
	In particular, 
	\begin{align*}
		\tau := \prob\left[T_w > \frac{\log n}{n}\right] &\leq \left(1 - \frac{\log n}{n} \right)^{d}
		\\ &= \left(1 - \frac{\log n}{n} \right)^{\frac{n}{\log n} (\log n \,-\, 2 \log\log\log n)}
		\\ &\leq e^{-(\log n \,-\, 2 \log \log \log n)}
		\\ &= \frac{(\log\log n)^2}{n} < 1 \,.
	\end{align*}
	
	Now let $T$ be the waiting time of the last vertex in $V \setminus S$ to become adjacent to at least one of the temporal sinks in $S$, i.e.,
	$T = \max \{ T_w ~|~ w \in V \setminus S \}$.
	Because the waiting times $T_w, w \in V \setminus S$ are all independent, we have
	\[
		\prob[T > x] = 1 - \prod_w \prob[T_w < x] \leq 1 - (1 - (1-x)^d)^r.
	\]
	In the rest of the proof we show that $\lim_{n \to \infty} \prob[T > \frac{\log n}{n}] = 0$.
	For this purpose, we derive
	\begin{align*}
		0 \geq r \cdot \log(1-\tau) 
		&\geq r \cdot \log \left(1 - \frac{(\log\log n)^2}{n}\right)
		\\ &\geq - 2r \frac{(\log\log n)^2}{n} \quad
		\text{[\small{since $(\log\log n)^2/n < 0.7$, assuming $n \geq 2$}]}
		\\ &= - 2 \frac{(\log\log n)^2 \cdot \log\log\log n}{\log n} \xrightarrow{n \to \infty} 0 \,.
	\end{align*}
	Thus, we have $\lim_{n\to\infty} r \cdot \log(1-\tau) = 0$ and therefore
	\begin{align*}
		\prob\left[T > \frac{\log n}{n}\right]
		 \leq 1 - ( 1 - \tau)^r
		 = 1 - \exp\left( r \cdot \log(1-\tau) \right)
		 \xrightarrow{n \to \infty} 1 - e^0 = 0 \,.
	\end{align*}
\end{proof}

\subsection{Temporal Spanners}
\label{sec:temporalSpanners}

The results from \cref{sec:TempSource} can be directly used to show that
$\frac{4 \log n}{n}$ is a sharp threshold for the existence of a pivotal spanner.
Indeed, let $\G \sim \F_{n,p}$ for $p \geq \frac{4 \log n}{n} + \frac{6(\log n)^{0.8}}{n}$,
and let $v$ be an arbitrary vertex in $\G$. Then, by \cref{th:avgSource},
with probability at least $1 - 5/\log n$ vertex $v$ is a temporal sink in
$\G_{[0,p_1]}$, and with at least the same probability $v$ is a temporal source in $\G_{[p_1,p]}$,
where $p_1 = \frac{2 \log n}{n} + \frac{3 (\log n)^{0.8}}{n}$.
Hence, by the union bound, with probability at least $1 - 10/\log n$ temporal graph $\G$
contains a pivotal temporal spanner formed as the union of the hindmost tree for $v$ in $\G_{[0,p_1]}$
and the foremost tree for $v$ in $\G_{[p_1,p]}$.
On the other hand, if $p  \leq \frac{4 \log n}{n} - \frac{6(\log n)^{0.8}}{n}$, then by \cref{th:firstSource}
for any $p_1 \in [0, p]$ with probability at least $1 - \frac{5}{\log n}$ either $\G_{[0,p_1]}$ has no temporal sink, or $\G_{[p_1,p]}$ has no temporal source (or both), and therefore $\G$ contains no pivotal temporal spanner.
Consequently, $\frac{4 \log n}{n}$ is a sharp threshold for the emergence of a Pivotal Temporal Spanner.

\begin{theorem}
	\label{th:pivotalSpan}
	The function $\frac{4\log n}{n}$ is a sharp threshold for the emergence of a Pivotal Temporal Spanner.
	More specifically, for any sufficiently large $n$,
	a random temporal graph in $\F_{n,p}$
	\begin{enumerate}[(i)]
		\item
		has no pivotal temporal spanner with probability at least $1 - \frac{5}{\log n}$,
		if $p\leq\frac{4\log n}{n}-\eps$; and
		\item
		has a pivotal temporal spanner with probability at least $1 - \frac{10}{\log n}$,
		if $p\geq\frac{4\log n}{n}+2\eps$,
	\end{enumerate}
	where $\eps := \frac{3(\log n)^{0.8}}{n} \in o\left(\frac{\log n}{n}\right)$.
\end{theorem}

While a pivotal temporal spanner contains only 2 more edges than an optimal spanner, the sharp threshold
for the existence of a pivotal temporal spanner is quite far from the temporal connectivity threshold.
This naturally raises a number of questions. For example, when does an optimal spanner appear? Or,
can we have a temporal spanner with a linear number of edges \aas when $p$ is close to the connectivity threshold?
We partially address the former question in \cref{sec:optimalSpanners} by showing that an optimal spanner exists \aas at $p=\frac{4\log n}{n}$; furthermore, we conjecture that $\frac{4\log n}{n}$ is a sharp threshold for the Optimal Temporal Spanner property.
In \cref{sec:crazy-spanner}, we provide a strong answer to the latter question, by proving that 
$\frac{3 \log n}{n}$ is a sharp threshold for the existence of a nearly optimal temporal spanner, i.e.,
a spanner with $(2+o(1))n$ time-edges.

The proofs of both results are different elaborations of the above simple construction of a pivotal temporal spanner.
The construction of a nearly optimal spanner is more involved among the two, and we develop it iteratively, by
first presenting in the next section a simpler construction of a temporal spanner with at most $5n$ edges at 
$p=3.5 \frac{\log n}{n}$.

\subsubsection{Warm-up: a spanner with $5n$ edges at $3.5  \frac{\log n}{n}$}
\label{sec:warmup}

In this section we show that a random temporal graph from $\F_{n,p}$ \aas contains a temporal spanner
with at most $5n$ edges if $p \geq (3.5+o(1)) \frac{\log n}{n}$.
The proof relies upon a construction of a pair of sparse subgraphs.
Each of these subgraphs consists of two edge-disjoint trees with a common root,
such that every vertex in one of the trees can reach the root before some time $t$, while
every vertex in the other tree can be reached from the root after $t$.
Thus all vertices of the first tree can reach all vertices of the second tree through the root.
The two subgraphs together provide temporal connectivity for all but a sublinear
number of pairs of vertices. Each of the remaining pairs is connected directly by a foremost path.

\begin{theorem}
	\label{th:span3_5}
	A random temporal graph from $\F_{n,p}$ \aas has a spanner with at most $(4+o(1))n$ edges
	if $p\geq 3.5 \frac{\log n}{n} + \frac{6 (\log n)^{0.8}}{n}$.
\end{theorem}
\begin{proof}
	Let $\G \sim \F_{n,1}$. Set $\eps := \frac{3 (\log n)^{0.8}}{n}$, and
	\begin{align*}
		r &:= 1.5\frac{\log n}{n} + \eps, &
		q &:= 2 \frac{\log n}{n} + \eps \,.
	\end{align*}
	
	\noindent
	Let $u$ and $v$ be two arbitrary vertices.
	By \cref{th:avgSource} and \cref{th:reachability} \ref{th:reachability-3} (with $z = 0.5$, $y = 1$),
	we have that each of the following events happens with probability at least 
	$1 - 5 / \log{n}$
	\begin{enumerate}
		\item $u$ is a temporal sink in $\G_{[0, q]}$; we denote by $T_{u}$ the hindmost tree for $u$ in $\G_{[0, q]}$.
		\item $v$ is a temporal source in $\G_{[r, p]}$; we denote by $T_{v}$ the  foremost tree for $v$ in $\G_{[r, p]}$.
		\item $u$ reaches at least $n - \frac{\sqrt{n}}{\log n}$ vertices in $\G_{[q, p]}$; we denote by $T_{u}'$ the foremost tree for $u$ in $\G_{[q, p]}$.
		\item $v$ is reached by at least $n - \frac{\sqrt{n}}{\log n}$ vertices in $\G_{[0, r]}$; we denote by $T_{v}'$ the hindmost tree for $v$ in $\G_{[0, r]}$.
	\end{enumerate}
	
	\noindent
	Clearly, these four trees together have at most $4(n-1)$ edges.
	Let $A$ be the set of vertices \emph{not} reaching $v$ through $T_v'$,
	and analogously let $B$ be the set of vertices \emph{not} reached by $u$ through $T_u'$.
	Since all edges of $T_u$ appear before time~$q$ while all edges of $T_u'$ appear later,
	$T_u \cup T_u'$ provides reachability for all ordered vertex pairs in $V \times (V \setminus B)$.
	Analogously, $T_v' \cup T_v$ provides reachability for all pairs in $(V \setminus A) \times V$.
	Thus it remains to provide reachability only for the pairs in $A \times B$.
	
	By \cref{th:tempConnectivity}, $\G_{[0,p]}$ is temporally connected \aas
	Hence, by \cref{cor:foremost-len-aas}, \aas for every ordered pair $(a,b) \in A \times B$ there exists a foremost $(a,b)$-path in $\G_{[0,p]}$ with fewer than $14 \log n$ edges.
	Therefore, in total these paths have at most 
	$\abs{A \times B} \cdot 14 \log n \leq \frac{n}{(\log n)^2} \cdot 14 \log n = \frac{14n}{\log n}$ edges.
	This collection of paths together with the four trees form a temporal spanner with at most $4(n-1) + \frac{14n}{\log n} = (4+o(1))n$ edges in $\G_{[0,p]}$ a.a.s.
\end{proof}

\subsubsection{Sharp threshold for nearly optimal temporal spanners}
\label{sec:crazy-spanner}

In this section we show that at time $3 \frac{\log n}{n}$,
which is when temporal connectivity emerges (\cref{th:tempConnectivity}),
a graph contains a temporal spanner of nearly optimal size $(2 + o(1))n$ \aas
To achieve this, we extend and refine the approach used in the previous section.
In particular, we give a more subtle construction of the two sparse subgraphs 
(\cref{lm:sink-spanner} and \cref{lm:source-spanner}) and conduct a more involved
accompanying analysis (\cref{lem:sparse-almost-spanner}) to show that the two subgraphs provide temporal connectivity for all but a sublinear number of pairs by time 
$(3+o(1)) \frac{\log n}{n}$. As before, any leftover pairs are connected via direct foremost paths.

Let $\G$ be a temporal graph and let $u$ and $v$ be two vertices in $\G$. 
The \emph{earliest arrival time} to $v$ from $u$ in $\G$, denoted by $\mu_\G(u,v)$, 
is the minimum time $t$ such that there exists a temporal $(u,v)$-path in $\G$ with arrival time $t$.
Similarly, the \emph{latest departure time} from $u$
to $v$ in $\G$, denoted by $\sigma_\G(u,v)$, is the maximum time $t$ 
such that there exists a temporal $(u,v)$-path in $\G$ with departure time $t$.
If $\G$ has no temporal $(u,v)$-path, we set $\mu_\G(u,v) = \infty$
and $\sigma_\G(u,v) = -\infty$. For notational convenience we define $\infty - \infty := 0$.
We proceed with the lemma showing that if $p$ is large enough, then
for any fixed vertex $u$ in $\G \sim \F_{n,p}$ there is \aas a sparse temporal subgraph
in which the earliest arrival time from any vertex $v$ to $u$ is almost the same as in $\G$.

\begin{lemma}
	\label{lm:sink-spanner}
	Let $\delta_0 :=\frac{28}{\log n}$ and $\eps_0 := \frac{6}{(\log n)^{0.2}}$.
	Then
	for any $p \geq(1+\eps_0)\frac{\log n}{n}$
	and any fixed vertex $u$ of $\G\sim{}\F_{n,p}$,
	\aas there is a temporal subgraph $\G'$ of $\G$
	with at most $(1 + \delta_0)n$ edges
	so that
	\begin{align*}
		\abs{\mu_\G(v,u) - \mu_{\G'}(v,u)} \leq \eps_0 \frac{\log n}{n}\end{align*}
	holds for all vertices $v \in V(\G)$.
\end{lemma}
\begin{proof}
	Let $u$ be an arbitrary vertex in $\G$.
	Set $q_+ := (1+ \eps_0/2)\frac{\log n}{n}$ and $q_- := (1- \eps_0/2)\frac{\log n}{n}$.
	By \cref{th:reachability} \ref{th:reachability-3} and \ref{th:reachability-2} (with $z=1$, $y=2$), vertex $u$
	is reachable by all but at most $\frac{n}{(\log n)^2}$ vertices in $\F_{n,q_{+}}$ a.a.s.,
	and vertex $u$ is reachable by at most $\frac{n}{(\log n)^2}$ vertices in $\F_{n,q_{-}}$ a.a.s.
	We construct $\G'$ as the union of
	\begin{enumerate}[(1)]
		\item the hindmost tree for $u$ in $\G_{[0, \,q_{+}]}$, and
		\item an earliest-arrival hindmost $(v,u)$-path in $\G$ of minimal length for each vertex~$v$
		with $\mu_\G(v,u) \notin [q_{-}, q_{+}]$. 
	\end{enumerate}
	Observe that $\mu_\G(v,u) \leq \mu_{\G'}(v,u) \leq q_{+}$ for all vertices $v$
	with $q_{-} \leq \mu_\G(v,u) \leq q_{+}$ and that $\mu_\G(v,u) = \mu_{\G'}(v,u)$
	for all other vertices $v$. 
	Thus $\abs{\mu_\G(v,u) - \mu_{\G'}(v,u)} \leq \eps_0 \frac{\log n}{n}$ holds for all vertices of $\G$.

	To estimate the size of $\G'$, we observe that the hindmost tree in (1) contains at most $n-1$ edges,
	and by the choice of $q_{-}$ and $q_{+}$, the number of paths added in (2) is \aas at most
	$\frac{2n}{(\log n)^2}$.
	Furthermore, by \cref{cor:foremost-len-aas}, \aas each of these paths has length at most $14\log n$.
	Thus, \aas $\G'$ has at most $n-1 + \frac{2n}{(\log n)^2} \cdot 14 \log n \leq (1+\delta_0)n$ edges, as required. 
\end{proof}

The following counterpart of \cref{lm:sink-spanner} holds by symmetry under time reversal and shows that,
for any fixed vertex $u$, \aas $\G$ contains a sparse temporal subgraph in which 
the latest departure times from $u$ to any vertex $v$ is almost the same as in $\G$.

\begin{lemma}
	\label{lm:source-spanner}
	Let $\delta_0 :=\frac{28}{\log n}$ and $\eps_0 := \frac{6}{(\log n)^{0.2}}$.
	Then for any $p\geq(1+\eps_0)\frac{\log n}{n}$
	and any fixed vertex $u$ of $\G\sim{}\F_{n,p}$,
	\aas there is a temporal subgraph $\G'$ of $\G$
	with at most $(1 + \delta_0)n$ edges
	so that
	\begin{align*}
		\abs{\sigma_\G(u,v) - \sigma_{\G'}(u,v)} \leq \eps_0 \frac{\log n}{n}
	\end{align*}
	holds for all vertices $v \in V(\G)$.
\end{lemma}

In the next lemma we combine \cref{lm:sink-spanner,lm:source-spanner} to construct a sparse spanner
in $\G \sim \F_{n,p}$, where $p = (3 + o(1))\frac{\log n}{n}$; the spanner provides temporal connectivity 
for all but at most a sublinear number of pairs of vertices.

\begin{lemma}\label{lem:sparse-almost-spanner}
	Let $\delta :=\frac{56}{\log n}$.
	There exists a function $\eps = \eps(n) \in o(1)$ such that for 
	$p = (3 + \eps)\frac{\log n}{n}$ 
	\aas a temporal graph $\G\sim\F_{n,p}$
	contains a temporal subgraph $\G'$ with at most $(2+\delta)n$ edges
	and there are at most $\frac{n\log\log n}{(\log n)^2}$ ordered pairs of vertices $(v,u)$ for which
	there is no temporal $(v,u)$-path in $\G'$.
\end{lemma}
\begin{proof}
	Let $\eps_0 := \frac{6}{(\log n)^{0.2}}$.
	We assume that $n$ is large enough so that $\eps_0 < 1$.
	By~\cref{th:avgSource} there is $\eps_1 = \eps_1(n) \in o(1)$ 
	such that an arbitrary vertex of $\F_{n,(2 + \eps_1) \frac{\log n}{n}}$ is a temporal source a.a.s.
	If necessary, we increase $\eps_1$ to ensure $(\eps_1 - \eps_0)\log n > 3 (\log n)^{0.8}$.	
	Set $r := \floor{\log\log n}$, $\eps := 2 \eps_1 +  2/r$, and
	\begin{align*}
		p_1  &:= (2 + \eps_1 + 1/r) \lnnf \\
		p_2  &:= p - p_1 = (1+\eps - \eps_1 -1/r) \lnnf = (1+\eps_1+1/r) \lnnf = p_1 - \lnnf 
	\end{align*}
	Fix an arbitrary vertex $u \in \G$.
	Note that, by choice of $\eps_1$ and as $p_1 > (2+\eps_1)\lnnf$,
	vertex $u$ is \aas temporal sink in $\G_{[0, p_1]}$
	and temporal source in $\G_{[p_2, p]}$.
	
	Define 	$\delta_0 := \frac{28}{\log n}$.
	By \cref{lm:sink-spanner}, there exists \aas a temporal subgraph $\G_1$ of $\G_{[0, p_1]}$ 
	of at most $(1 + \delta_0)n$ edges such that for all vertices $v \in V(\G)$
	\[
	\abs{\mu_\G(v,u) - \mu_{\G_1}(v,u)} \leq \eps_0 \frac{\log n}{n} \,.
	\]
	In particular, $u$ is a temporal sink in $\G_1$ and thus every vertex $v$ has $\mu_{\G_1}(v, u) \leq p_1$.
	Symmetrically, by \cref{lm:source-spanner}, there exists \aas a temporal subgraph $\G_2$ of $\G_{[p_2, p]}$
	of at most $(1 + \delta_0)n$ edges such that for all vertices $w \in V(\G)$
	\[
	\abs{\sigma_\G(u,w) - \sigma_{\G_2}(u,w)} \leq \eps_0 \frac{\log n}{n}
	\]
	and in particular $\sigma_{\G_2}(u, w) \geq p_2$ for every vertex $w$, because $u$
	is a temporal source in $\G_{[p_2, p]}$.
	
	We define $\G'$ as the union $\G_1 \cup \G_2$ and observe that $G'$ has at most
	$(2 + 2\delta_0)n = (2 + \delta)n$ edges.
	Note that there exists a temporal $(v,w)$-path in $\G_1 \cup \G_2$ if $\mu_{\G_1}(v, u) < \sigma_{\G_2}(u, w)$.
	Otherwise, \ie if $\mu_{\G_1}(v, u) \geq \sigma_{\G_2}(u, w)$, we call the ordered pair $(v, w)$ \emph{mismatched}.
	We claim that the total number of mismatched pairs, and therefore the number of pairs of vertices
	that are not temporally connected in $\G'$, is at most $r \cdot \frac{n}{(\log n)^2} \leq \frac{n \log\log n}{(\log n)^2}$.
	To this end, we introduce, for $0 \leq i \leq r$,
	\begin{align*}
		q_i &:= p_2 + i \cdot \frac{p_1 - p_2}{r} = p_2 + \frac{i}{r} \cdot \frac{\log n}{n} \\
		A_i &:= \left\{v \in V(\G) \mid \mu_{\G_1}(v, u) \geq q_i \right\}\\
		B_i &:= \left\{w \in V(\G) \mid \sigma_{\G_2}(u, w) \leq q_i \right\}
		\intertext{and, for $0 \leq i < r$,}
		C_i &:= A_i \times B_{i+1} \,.
	\end{align*}
	Note that each mismatched pair is contained in some $C_i$. Indeed, a mismatched pair $(v, w)$
	belongs to $C_i$ for the maximum $i \in \{ 0, 1, \ldots, r-1 \}$ such that $q_i \leq \mu_{\G_1}(v,u)$.
	For every vertex $v \in A_i$, $0 \leq i < r$, we have
	\begin{align*}
		\mu_{\G}(v, u) 
		&\geq \mu_{\G_1}(v, u)  - \eps_0 \frac{\log n}{n} 
		\\ &\geq q_i - \eps_0 \frac{\log n}{n}
		\\ &= p_2 + \left(\frac{i}{r} -\eps_0 \right) \frac{\log n}{n}
		\\ &= \left( 1 + \eps_1 + \frac{1}{r} - \eps_0 + \frac{i}{r} \right) \cdot \frac{\log n}{n}
		\\ &= \left( 1 + \eps_1 - \eps_0 +  \frac{i+1}{r} \right) \cdot \frac{\log n}{n}
		\\ &> \left( 1 + \frac{i+1}{r} \right) \cdot \frac{\log n}{n} + \frac{3(\log n)^{0.8}}{n} =: a_i.
	\end{align*}
	
	\noindent
	On the other hand, by \cref{th:reachability} \ref{th:reachability-3} (with $z=1- (i+1)/r$ and $y=1$),
	with probability at least $1 - 5 / \log{n}$, at least $n-\left\lfloor \frac{n^{1- (i+1)/r}}{\log n} \right\rfloor$ vertices in $\G_{[0,a_i]}$ 
	can reach $u$, and therefore $\mu_{\G}(v, u) \leq a_i$ holds for all those vertices.
	Consequently, by the union bound, with probability at least $1 - 5r/ \log{n}$, we have $|A_i| \leq \frac{n^{1- (i+ 1)/r}}{\log n}$
	for every $0 \leq i < r$.
	
	Similarly, for every vertex $w \in B_i$, $0 < i \leq r$, we have 
	\begin{align*}
		\sigma_{\G}(u, w) 
		&\leq \sigma_{\G_2}(u, w)  + \eps_0 \frac{\log n}{n} 
		\\ &\leq q_i + \eps_0 \frac{\log n}{n}
		\\ &= p_2 + \left(\frac{i}{r} + \eps_0 \right) \frac{\log n}{n}
		\\ &= p - \left( 2 + \eps_1 + \frac{1}{r} - \eps_0 - \frac{i}{r} \right) \frac{\log n}{n} 
		\\ &< p - \left( \left( 2 - \frac{i-1}{r}  \right) \frac{\log n}{n} + \frac{3(\log n)^{0.8}}{n} \right) =: p - b_i.
	\end{align*}	
	
	\noindent
	On the other hand, by \cref{th:reachability} \ref{th:reachability-3} (with $z=(i-1)/r$ and $y=1$),
	with probability at least $1 - 5 / \log{n}$, at least $n-\left\lfloor \frac{n^{(i-1)/r}}{\log n} \right\rfloor$ vertices in $\G_{[p-b_i,p]}$
	can be reached from $u$, and therefore $\sigma_{\G}(u, w)  \geq p - b_i$ holds for all those vertices.
	Consequently, by the union bound, with probability at least $1 - 5r/ \log{n}$, we have $|B_i| \leq \frac{n^{(i-1)/r}}{\log n}$
	for every $0 < i \leq r$.

	Thus \aas  for all $0 \leq i < r$
	\[
	\abs{C_i} = \abs{A_i} \times \abs{B_{i+1}} \leq 
	\frac{n^{1- (i + 1)/r}}{\log n} \cdot \frac{n^{i/r}}{\log n} 
	= \frac{n^{1-1/r}}{(\log n)^2} < \frac{n}{(\log n)^2},
	\]
	and hence, as claimed, overall there are at most $\sum_{i=0}^{r-1} \abs{C_i} \leq r \cdot \frac{n}{(\log n)^2}$ 
	mismatched pairs.
\end{proof}

We are now ready to state and prove our main result.

\begin{theorem}
	\label{th:main}
	Let $\delta := \frac{56}{\log n} + \frac{14 \log\log n}{\log n} \in o(1)$.
	Then $p = 3\frac{\log n}{n}$ is a sharp threshold for a random temporal graph $\G \sim \F_{n,p}$ to have a spanner of size at most $(2 + \delta)n$.
\end{theorem}
\begin{proof}
	Let $\eps'(n)$ be given by \cref{lem:sparse-almost-spanner},
	and $\eps''(n)$ be given by \cref{th:tempConnectivity}.
	Set $\eps := \eps'(n)+\eps''(n) \in o(1)$.
	If $p \leq (3 - \eps) \frac{\log n}{n}$, then by \cref{th:tempConnectivity} \aas a random temporal graph
	$\G \sim \F_{n,p}$ is not temporally connected, and therefore it has no spanner at all.
	Hence, to prove the theorem,we show in the rest of the proof that if $p \geq (3 + \eps) \frac{\log n}{n}$,
	then \aas $\G \sim \F_{n,p}$ has a spanner of size at most $(2 + \delta)n$.
	
	Let $\G'$ be a subgraph of $\G$ with at most $\left( 2 + \frac{56}{\log n} \right)n$ edges as guaranteed by \cref{lem:sparse-almost-spanner}.
	Then \aas $\G'$ contains at most $\frac{n\log\log n}{(\log n)^2}$ ordered pairs $(v, w)$ with no temporal $(v,w)$-path.
	Since $\G$ is \aas temporally connected by \cref{th:tempConnectivity},
	it \aas contains a temporal $(v, w)$-path for each of these pairs $(v, w)$,
	and by \cref{cor:foremost-len-aas} \aas each of these temporal paths can be chosen to have length at most $14 \log n$.
	By adding these paths to $\G'$, we obtain a temporal spanner of size at most
	\[
	\left( 2 + \frac{56}{\log\log n} \right)n + \frac{n\log\log n}{(\log n)^2} \cdot 14 \log n = (2+\delta)n.
	\]
\end{proof}

\subsubsection{Optimal Temporal Spanners}
\label{sec:optimalSpanners}

In this section we show that a random temporal graph
from $\F_{n,p}$ contains an optimal spanner (\ie a spanner with exactly $2n-4$ edges) \aas
whenever $p \geq (4+o(1))\frac{\log n}{n}$.
The main idea of the construction is to replace the pivot vertex in
the pivotal temporal spanner construction by a pivot $4$-cycle.

\begin{theorem}
\label{th:opt-span4}
A random temporal graph in $\F_{n,p}$ \aas has an optimal spanner
if $p \geq 4\frac{\log n}{n} + \eps$
where $\eps := \frac{16 (\log n)^{0.8}}{n} \in o$$\left(\frac{\log n}{n}\right)$.
\end{theorem}

\begin{proof}
Assume $n \geq 8$ and set
$\eps_0 := \frac{\eps}{4} = \frac{4 (\log n)^{0.8}}{n} > \frac{3 (\log (n-3))^{0.8}}{n-3} > \frac{\log \log n}{n}$.
Let also $\G \sim \F_{n,1}$.
Our goal is to show that $\G_{[0, p]}$ has an optimal spanner a.a.s.
For this purpose, set
\begin{align*}
	p_1 &:= 2\frac{\log n}{n} + \eps_0,
	\\ p_2 &:= 2 \frac{\log n}{n} + 2\eps_0,
	\\ p_3 &:= 2 \frac{\log n}{n} + 3 \eps_0 \,.
\end{align*}
Define $V^{(4)}$ to be the set of all 4-tuples of pairwise distinct vertices of $\G$.
We say that such a tuple $(w, x, y, z) \in V^{(4)}$ of $\G$ forms a \emph{square}
if $\{\lambda(wx), \lambda(yz)\} \subset [p_1, p_2]$
and $\{\lambda(xy), \lambda(wz)\} \subset [p_2, p_3]$.
Note that in a square each vertex can reach any other by a temporal path with labels from $[p_1, p_3]$.
Our first goal is to show that \aas there exist sufficiently many squares.
This will guarantee that at least one of them can be chosen as a pivot $4$-cycle.

Let $S$ be the number of squares in $\G$.
Clearly $\expect[S] = n^{(4)}\eps_0^4$,
where $n^{(i)} := n\cdot(n-1)\cdot(n-2) \dotsm  (n-i+1)$ denotes the falling factorial.
To bound the variance of $S$, we need to investigate
\[
	\expect[S^2] = \sum_{A \in V^{(4)}} \sum_{B \in V^{(4)}} \prob[A, B \text{ are both squares}] 
	= \sum_{i=4}^8 \sum_{\substack{A, B \in V^{(4)}\\\abs{A \cup B} = i}} \prob[A, B \text{ are both squares}] \,.
\]
For $i = 8$ (i.e., $A, B$ being vertex-disjoint) we clearly have $n^{(8)}$ summands,
whereas for $i \leq 7$ the number of pairs $A, B \in V^{(4)}$ with $\abs{A \cup B} = i$ can be upper-bounded by $(8n)^{i} \leq 8^7 n^{i}$.
The time label of every edge in $A \cup B$ belongs to an interval of length $\varepsilon_0$ and thus $\prob[A, B \text{ are both squares}] \leq \eps_0^i$.
Therefore,
\[
	\expect[S^2] \leq n^{(8)} \eps_0^8 + \sum_{i=4}^7 8^7 n^i \eps_0^i
	\leq \expect[S]^2 + 8^7 \cdot \sum_{i=4}^7 4^i (\log n)^{0.8i}
	\leq \expect[S]^2 + 8^7 \cdot 4 \cdot 4^7 (\log n)^{5.6}
\]
and thus the variance $\variance[S]$ is no more than
        $8^7\cdot 4^8(\log n)^{5.6} \in o(\expect[S]^2)$.
By Chebyshev's inequality, we obtain that $S \geq \expect[S]/2$ a.a.s.

Assume now a 4-tuple~$(w, x, y, z)$ of vertices to be given
(which might or might not form a square).
Let $\G'$ be the temporal subgraph of $\G$
obtained by deleting the vertices $x, y, z$.
Observe that $\G'$ is distributed as an element of $\F_{n-3,\:1}$,
independently of whether $(w,x,y,z)$ form a square in $\G$.
Note that the intervals $[0, p_1]$ and $[p_3, p]$ both have length $\frac{2 \log n}{n} + \eps_0$.
Thus, by \cref{th:avgSource} and the union bound, with probability at least $1 - \frac{10}{\log (n-3)}$,
vertex $w$ is a temporal sink of $\G'_{[0, p_1]}$
and also a temporal source in $\G'_{[p_3, p]}$.
If $(w,x,y,z)$ additionally form a square in $\G$,
then we call $(w,x,y,z)$ a \emph{good square}.
Then, by taking the good square $(w,x,y,z)$
together with the foremost tree for $w$ in $\G'_{[p_3, p]}$ and the hindmost tree for $w$ in $\G'_{[0, p_1]}$,
we obtain a spanner of $\G_{[0, p]}$ that has $4 + 2(n-4) = 2n - 4$ edges.
\newcommand{\Sbad}{S_{\textnormal{bad}}}

It only remains to show that there exists a good square.
According to the above, the expected number of bad (i.e., non-good) squares $\Sbad$ satisfies $\expect[\Sbad] \leq \expect[S] \frac{10}{\log (n-3)}$.
By Markov's inequality and our bound on $S$, the probability for every square of $\G$ to be bad is at most
\[
	\prob[\Sbad = S] 
\leq \prob\left[S < \frac{\expect [S]}{2}\right] +
\frac{\expect[\Sbad]}{(\expect S) / 2}
 \leq 
o(1)+
\frac{20}{\log(n-3)} \in o(1)\,.
\]
Thus, $\G$ has a good square a.a.s.
\end{proof}

We conjecture that $4\frac{\log n}{n}$ is in fact a sharp threshold for the existence of an optimal spanner.

\begin{conjecture}\label{conjecture}
The sharp threshold for the existence of an optimal spanner in a random simple temporal graph
exists and is equal to $4\frac{\log n}{n}$.
\end{conjecture}

\section{Connections to other models}
\label{sec:other-models}

In this section, we explore some connections to other models in the literature.
First and most importantly, we show how our results translate to gossiping and population protocols, with or without repetition of interactions between pairs of nodes.
Afterwards, we look into quite surprising parallels to known results about randomly weighted static graphs.
Finally, we mention existing results in random edge-ordered graphs,
where the model is equivalent to RSTGs but the questions that have been studied are different.

\subsection{Gossiping and population protocols}
\label{sec:gossip-population}

In the classical gossiping setting, there are $n$ agents, each of whom
knows a single secret. The agents can communicate via telephone calls.
Whenever an agent calls another agent, the two exchange
all the secrets they know. An agent who learns all $n$ secrets becomes 
an \emph{expert}. The standard goal is to reach a configuration in which all agents are experts. A sequence of calls leading to this state is called a \emph{gossiping protocol}\footnote{There are related models of \emph{parallel rumor spreading} in which multiple pairwise calls can occur in a single round. In those models, the target measure is usually the number of rounds needed for secrets spreading (see, e.g., \cite{frieze1985shortest, karp2000randomized}).}.

In the theory of population protocols, a similar process appears in the analysis of time complexity. Here the order of interactions is governed by a sequential \emph{scheduler}. Unless this scheduler is adversarial, it is generally assumed, for analysis, that the interactions between pairs of nodes are chosen uniformly at random, with repetitions.
The complexity of a protocol is then measured in terms of the number of interactions required to reach a certain configuration, which depends on the problem at hand.

If we consider the population protocol's interactions to be symmetric (\ie not oriented), then the above two models are equivalent in terms of interactions.
For the sake of concreteness, in what
follows, we will rely on the terminology of gossiping protocols (in particular~\cite{van2017reachability}). 
Two natural models of random interactions can be defined, depending on whether repetitions are allowed between the same pair of agents:
\begin{enumerate}
	\item[\ANY{}:] No restrictions apply; every call is performed uniformly at random, without dependency on previous calls;
	
	\item[\CO{}:] Stands for \emph{call-once}, \ie no call can be repeated; there is at most one call between any pair of agents.
\end{enumerate}

According to a classical result from the '70s, gossiping protocols
require at least $2n-4$ calls even deterministically (see \cite{bollobas2006art} for a historical 
note on the corresponding combinatorial problem).
On the probabilistic side, the duration of gossiping protocols has
been studied mostly in the \ANY{} model.
In a sequence of three papers
\cite{moon1972random, boyd1979random, haigh1981random}
with the same title ``\emph{Random exchanges of information}'',
asymptotics for the \emph{expected} number of calls until a fixed agent becomes an expert and
until all agents become experts were obtained for the randomized \ANY{} model
of communication, i.e., for the model where every
next call happens between a pair of distinct agents chosen uniformly
at random among all pairs of distinct agents.

\begin{theorem}[\cite{moon1972random, boyd1979random, haigh1981random}]
	In the randomized \ANY{} model of communication,
	\begin{enumerate}
		\item the expected number of calls until a \emph{fixed} agent becomes an expert is
		$n \log n + \bigO(n)$;
		\item the expected number of calls until \emph{all} agents become experts is
		$\frac{3}{2} n \log n + \bigO(n)$.
	\end{enumerate}
\end{theorem}

Unsurprisingly, the number of calls until a fixed agent becomes an expert also turns out to be concentrated around its expected value. To our knowledge, the best concentration to date was obtained in~\cite{Emmanuelle}. We extract only the following:

\begin{theorem}[\cite{Emmanuelle, Mocquard2019}]
	In the randomized \ANY{} model of communication, \aas the number of calls until a \emph{fixed} agent becomes an expert is 
	$ (1 + o(1)) (n \log n)$.
\end{theorem}

As for the number of calls until all agents become experts, it is known that it does not exceed its expected value \aas. The best bound to date appears to be from~\textcite{burman2021time}.
We remark that the actual lemma in~\cite{burman2021time}
states a weaker bound on time (to achieve a stronger bound on the probability),
but an intermediate step of the proof directly implies the following:

\begin{theorem}[\cite{burman2021time}, proof of Lemma~2.9]
	In the randomized \ANY{} model of communication, a.a.s.
	the number of calls until \emph{all} agents become experts is \emph{at most}
	$(1 + o(1)) (\frac{3}{2} n \log n)$ a.a.s.
      \end{theorem}

\noindent
According to \textcite{van2017reachability}, there seem to be no similar estimations known
for the randomized \CO{} model, which may arguably be due to the fact that the no-repetition property of this model creates dependencies between past and future events which are more difficult to handle.
In \cref{sec:CO} (below), we show how our results readily translate to the \CO{} model. Quite significantly, we provide asymptotic estimates
for the actual number of calls, not only the expected value.
In \cref{sec:ANY}, we demonstrate the flexibility of our results by showing that they apply to the \ANY{} model as well, strengthening the known results in gossiping theory and in the analysis of population protocols.
Indirectly, the transferability of our results also implies that information propagates at essentially the same speed with or without repetitions in the model.

Finally, we do not discuss our results on temporal spanners in this section, although they could be translated in a similar way, because the corresponding problems have not been considered in the areas of gossiping protocols and population protocols. The remainder of this section provides technical details on the claimed equivalences.

\subsubsection{Random information exchanges in the CO model}
\label{sec:CO}

Recall that in our model $\F_{n,1}$, the time labels induce edge orderings of the complete graph,
each of which is equiprobable.
The same distribution of edge orderings is obtained if we construct 
a random ordering by choosing edges one by one uniformly at random among all
edges that are not yet chosen, i.e., in the same way as calls are originated in the
randomized \CO{} model.
Therefore, by interpreting edges as calls and time labels as ranks of the calls,
the number of edges in $\G_{[0, p]}$, where $\G \sim \F_{n,1}$, is the number of calls that 
happened no later than time $p$. 
Furthermore, since the underlying graph of $\G_{[0, p]}$ is distributed according to
$G_{n,p}$ and the number of edges in $G_{n,p}$ is concentrated around ${n \choose 2} p$
when $p = \Theta(\log n / n)$, our results about random temporal graphs from \cref{sec:tempProp}
translate straightforwardly to the following results in the randomized \CO{} model.

\begin{theorem}\label{th:CO}
	In the randomized \CO{} model, \aas
	\begin{enumerate}
                \item the number of calls until two fixed agents exchange their secrets is $\frac{1}{2} n \log n \cdot (1 + o(1))$;
                        \\ \qquad (\cref{th:pathUV})
		\item the number of calls until at least one agent becomes an expert is $n \log n \cdot (1 + o(1))$;
                        \\ \qquad (\cref{th:firstSource})
		\item the number of calls until a fixed agent becomes an expert is $n \log n \cdot (1 + o(1))$;
                        \\ \qquad (\cref{th:avgSource})
		\item the number of calls until all agents become experts is $\frac{3}{2}n \log n \cdot (1 + o(1))$.
                        \\ \qquad (\cref{th:tempConnectivity})
	\end{enumerate}
\end{theorem}

\subsubsection{Random information exchanges in the ANY model}
\label{sec:ANY}

By replacing the uniform distribution of edge labels with another distribution, we can
obtain random temporal graphs with different temporal dynamics.
This flexibility, in particular, allows us to simulate the randomized \ANY{} model. In order to achieve this, we replace the uniform distribution for edge labels in $\F_{n,1}$ 
with a suitable Poisson point process.

More specifically, we introduce a random temporal graph model $\Hh_n$,
in which the labels of each of the ${n \choose 2}$ potential edges appear independently according to a Poisson point process 
with rate $1$ starting at time~$0$ and running infinitely long.
In this model each edge gets a countably infinite set of labels with probability~$1$.
We also define a finite random temporal graph model $\Hh_{n,p}:=(\Hh_n)_{[0,p]}$
where the process is stopped at time $p$.
Note that, in $\Hh_{n,p}$, an edge may appear an arbitrary number of times, 
including zero in which case the edge is not present in the underlying graph. 
The expected number of appearances for each edge is exactly $p$.
Furthermore, as the Poisson distribution has variance equal to the expectation, by Chebyshev's inequality the number of appearances of a fixed edge is concentrated around its expected value.
Hence, the total number of edge appearances in $\Hh_{n,p}$ 
is concentrated around ${n \choose 2} p$.

Since all edges appear according to independent identical Poisson point processes,
at any fixed point in time, the next edge to appear is distributed uniformly at random among all $\binom{n}{2}$ possible edges, i.e., 
in the same way as calls are scheduled in the randomized \ANY{} model.
Hence, if $p_0$ is a threshold probability for a temporal property in $\Hh_{n,p}$, then it translates to
a ${n \choose 2} p_0$-calls threshold in the randomized \ANY{} model for the corresponding property.

In this section, we will show that all our main results about $\F_{n,p}$ can be transferred to $\Hh_{n,p}$ with only minor changes in the proofs.
Precisely, the corresponding results from \cref{sec:tempProp}
translate to the following results in the randomized \ANY{} model.

\begin{theorem}\label{th:ANY}
	In the randomized \ANY{} model, \aas
	\begin{enumerate}
		\item the number of calls until two fixed agents exchange their secrets is $\frac{1}{2} n \log n \cdot (1 + o(1))$;
		\item the number of calls until at least one agent becomes an expert is $n \log n \cdot (1 + o(1))$;
		\item the number of calls until a fixed agent becomes an expert is $n \log n \cdot (1 + o(1))$;
		\item the number of calls until all agents become experts is $\frac{3}{2}n \log n \cdot (1 + o(1))$.
	\end{enumerate}
\end{theorem}

In the remainder of this section, we explain which changes (if any) need to be made to the statements of
\cref{sec:2hop,sec:foremostTree,sec:tempProp} in order to apply to the $\Hh_{n,p}$ model.

\subsubsection*{Results based on 2-hop approach (\cref{sec:2hop})}

\begin{lemma}[cf.\ \cref{thm:2hop-result}]
\label{thm:2hop-result-ANY}
	Let $\alpha \geq 3$ and let $p = \alpha\sqrt{\log{n}/n}$.
	Then, for all $n \geq 4$ and $p\leq 1$,
an arbitrary vertex of $(G, \lambda) \sim \Hh_{n,p}$ is a temporal source 
	with probability at least $1 - n^{-\alpha^2/12 + 1}$.
\end{lemma}
\begin{proof}
The construction with intermediate vertices stays the same.
Note that the number of edges of a fixed path~$x,y,z$ that one can traverse in~$\Hh_n$ before some point in time
can be modeled as a Poisson process stopped after two occurences.
Thus, the probability~$\prob[S_{yz}]$ that vertex $x$ can reach $z$ along said path before time~$p$ becomes
$1-e^{-p}(1+p)
>1-\exp(-p+p-p^2/2+p^3/3)
=1-\exp(-p^2/2+p^3/3)
>1-\exp(-p^2/6)$.
Here the first inequality uses the Taylor series of the natural logarithm,
and the second uses $p\leq 1$.
The probability $p_1$ of all $n-2$ intermediate vertices being unsuitable
is at most 
$(1-(1-\exp(-p^2/6)))^{n-2}
=\exp(-(p^2/6)(n-2))
=\exp(-\alpha^2 (\log n)(n-2)/(6n))
\leq n^{-\alpha^2/12}$
by independence.
The rest of the proof is the same as before with $\alpha^2/12$ instead of $\alpha^2/4$.
\end{proof}

Just like \cref{cor:2hopTempConn}, we obtain the following corollary.

\begin{corollary}[cf.\ \cref{cor:2hopTempConn}]\label{cor:2hopTempConn-ANY}
	Let $p = \frac{\log{n}}{\sqrt{n}}$. Then, $(G, \lambda) \sim \Hh_{n,p}$ is temporally connected with probability at least $1-n^{-\frac{\log{n}}{12} + 2}$.
\end{corollary}

\subsubsection*{Results on foremost tree evolution (\cref{sec:foremostTree})}

A change in the foremost tree algorithm (\cref{alg:foremostTree}) is necessary
to account for the multiple labels on some edges.
The algorithm is now applied to the graphs distributed 
according to $\Hh_n$.
The change is similar to treating each edge label as a separate edge.
More specifically, for each edge $e$ in $S_k$, i.e., an edge connecting two vertices,
exactly one of which is included in the tree $T_{k-1}$,
we consider the earliest label $\lambda_{e,k}\in\lambda(e)$ such that $T_{k-1}\cup\{(e,\lambda_{e,k})\}$
is an increasing temporal tree. 
With probability $1$, such a label exists.
Just like before, we select the edge $e$ with the minimal possible $\lambda_{e,k}$.

Apart from the fact that $(G, \lambda)$ is now sampled from $\Hh_{n}$
and notation changes to use $\lambda_{e,k}$ instead of $\lambda(e)$ where appropriate,
the definitions of $T_k^v$, $Y_k^v$, $c_k$, $\capped{Y}_k^v$ and $\A_k^v$ apply without changes,
as do \cref{lem:foremostTree} and \cref{lm:XCapAAS_sum}:

\begin{lemma}[cf.\ \cref{lem:foremostTree}]
        Let $(G, \lambda)$ be a temporal graph (not necessarily simple)
        and $v$ be a temporal source in $(G, \lambda)$. Then
	\begin{enumerate}[(i)]
		\item the modified \cref{alg:foremostTree} constructs a foremost tree for $v$ in $(G, \lambda)$;
                \item $\lambda_{e_1,1} \leq \lambda_{e_2,2} \leq \cdots \leq \lambda_{e_{n-1},n-1}$.
	\end{enumerate}
\end{lemma}
\begin{proof}
Apart from notation changes, this lemma is proven exactly as before.
\end{proof}

Instead of \cref{lem:ExpOfXk}, we may state a slightly stronger result here which also replaces \cref{cor:ExpOfcappedXk}:

\begin{lemma}[cf.\ \cref{lem:ExpOfXk}]\label{lem:ExpOfXk-ANY}
For a vertex $v$, with probability at least $1 - 4/\log{n}$ the equality $\capped{X}^v_k=X^v_k$ holds for every $k \in [n-1]$.

Further, for a vertex $v$ and every $k \in [n-1]$, we have
	\begin{flalign*}
		 \text{(i)} && \expect[X^v_k \mid \A^v_{k-1}]  &= \frac{1}{k(n-k)}; &
		 \\
		 \text{(ii)}&& \frac{1 - 1/ \log n}{k(n-k)} \leq \expect[\capped{X}^v_k \mid \A^v_{k-1}] &\leq \frac{1}{k(n-k)}. &
	\end{flalign*}
\end{lemma}
\begin{proof}
	The proof of the first part of the lemma is exactly the same as the proof of corresponding claim
	in \cref{lem:ExpOfXk} with the only 
difference that $e^{-k(n-k)c_{k}}$ 
(from the computation in \cref{lm:XCapAAS_sum}) is now the exact value of $\prob(X^v_k\geq c_k)$ rather than an upper bound on it.

	(i):
	Due to the memorylessness of the exponential distribution, $X^v_k$ conditioned on $\A^v_{k-1}$ is distributed as the minimum of $k(n-k)$ independently Exp(1)-distributed random variables.
	Thus, the distribution of $X^v_k$ is exactly Exp($k(n-k)$) and as such it has expected value $\frac{1}{k(n-k)}$.

	(ii):
	The only difference to the proof of \cref{lem:ExpOfXk} is in the evaluation of the integral
	\begin{align*}
	\int_{0}^{c_k} \prob[X^v_k \geq t \mid A^v_{k-1}] \dd t
			&=
			\int_{0}^{c_k} \prob[X_k^v \geq t] \dd t 
			\\ &=
			\int_{0}^{c_k} \exp\left( -k (n-k) t \right) \dd t
			\\ &=
			\frac{1 - \exp(-k(n-k)c_k)}{k(n-k)} \,.
	\end{align*}
	Then, the desired bound follows from the fact that
   $\exp(-k(n-k)c_k)  \leq  \frac{1}{\log n}$ (see (\ref{eq:prob_capped_uncapped_sum})).
\end{proof}

The remaining results in \cref{sec:foremostGrowth} (\cref{lm:YCapConcentrate} to \cref{th:reachability}) and their proofs carry over to $\Hh_{n,p}$ with only trivial modifications.

Regarding \cref{sec:foremostTreeHeight},
we observe that \cref{lm:foremost-heights-step-dist} holds trivially for $\Hh_{n}$, 
as the attachment point is uniformly distributed for the Poisson point process case.
Consequently, \cref{lm:foremost-height-step} and \cref{lm:foremost-height-upper-bounds} also hold,
and so do \cref{the:foremost-heights-aas} and \cref{cor:foremost-len-aas}
as all the properties of $\F_{n,p}$ used in the proofs are shared by $\Hh_{n,p}$.

\subsubsection*{Results on sharp thresholds for temporal graph properties (\cref{sec:tempProp})}

\cref{th:pathUV,th:avgSource,th:firstSource} 
(sharp thresholds for point-to-point reachability, temporal source, and first temporal source)
apply without modifications as they follow solely from the concentration results obtained in \cref{sec:foremostGrowth}.
\Cref{lm:tconnLower} only requires a minor change to the proof:
\begin{lemma}[cf.\ \cref{lm:tconnLower}]
\label{lm:tconnLower-ANY}
Let $p \leq \frac{3\log n}{n} -\eps$, where $\eps := \frac{6 (\log n)^{0.8}}{n}$,
and let $\G \sim \Hh_{n}$.
Then, \aas the temporal graph $\G_{[0, p]}$ contains at least one vertex which is not a temporal sink.
\end{lemma}
\begin{proof}
The proof works the same,
except that the probability for a pair of vertices to form an edge in the underlying graph $H$ of $\G_{[q,p]}$ is always
\[
	\gamma = 1 - e^{q-p} \leq p - q 
	\leq \frac{3 \log n}{n} - \eps - \frac{2 \log n}{n} + \frac{\eps}{2}
	< \frac{\log n - (\log n)^{0.8}}{n} \,.
\]
In particular, $H$ is distributed as an element of $G_{n, \gamma}$
and the remainder of the proof may be copied as is.
\end{proof}
The proof of \cref{lem:connUpperBound} is simplified due to the fact that we do not need to bound $\Delta(\G_{[0, q]})$ anymore.
\begin{lemma}[cf.\ \cref{lem:connUpperBound}]
	\label{lem:connUpperBound-ANY} 
	Let $p \geq \frac{3\log n}{n} + \eps$, where $\eps := \frac{3 (\log n)^{0.8}}{n}$,
	and let $\G \sim \Hh_{n}$.
	Then, \aas every vertex in $\G_{[0, p]}$ is a temporal sink.
\end{lemma}
\begin{proof}
We define $q$ and $r$ and $S$ as in \cref{lem:connUpperBound},
and observe that $\abs{V \setminus S} \leq r$ a.a.s.
The set $C_w$ of potential edges connecting $w \in V \setminus S$ to an element of $S$ is simplified to
\[
 C_w = \{vw \mid v \in S\}
\]
which now trivially implies $\abs{C_w} = n - 1 - r \geq n - 2r = d$.
As the waiting time of each edge is exponentially distributed, we now get
$\prob[T_w > x] \leq e^{-xd}$.
In particular, the estimation
\[
 \tau := \prob\left[T_w > \frac{\log n}{n}\right] \leq e^{-(\log n \,-\, 2 \log \log \log n)}
\]
still applies, so the remainder of the proof may be copied without modifications.
\end{proof}
Thus, we again obtain \cref{th:tempConnectivity}, i.e., the sharp threshold
on temporal connectivity in $\Hh_{n,p}$.

Of our results in \cref{sec:temporalSpanners}, all except \cref{th:opt-span4} are based solely on results which we have already proven to also hold for $\Hh_{n,p}$.
Finally, \cref{th:opt-span4} only requires a minor modification to the proof.

\begin{theorem}[cf.\ \cref{th:opt-span4}]
\label{th:opt-span4-ANY}
$\Hh_{n,p}$ \aas has an optimal spanner
if $p \geq 4\frac{\log n}{n} + \eps$
where $\eps := \frac{16 (\log n)^{0.8}}{n} \in o\left(\frac{\log n}{n}\right)$.
\end{theorem}
\begin{proof}
In comparison to the proof of \cref{th:opt-span4},
only the probability of an edge appearing in an interval of length $\eps_0$
changes from $\eps_0$ to $1 - e^{-\eps_0}$.
Thus, $\expect[S] = n^{(4)}(1-e^{-\eps_0})^4$
and
\begin{align*}
	\expect[S^2] &\leq n^{(8)} (1-e^{-\eps_0})^8 + \sum_{i=4}^7 8^7 n^i (1-e^{-\eps_0})^i
	\leq \expect[S]^2 + 8^7 \cdot \sum_{i=4}^7 n^i \eps_0^i \,.
\end{align*}
Therefore, the bound on the variance $\variance[S]$ and the remainder of the proof remain valid.
\end{proof}

\subsection{Complete graphs with random weights}
\label{sec:first-passage-percolation}

Another model related to RSTGs (technically, to $\F_{n,p}$), although less closely than gossiping and population protocols scheduling,
is the randomly edge-weighted complete graph model studied by \textcite{Janson99}
and related to percolation theory.
Janson's model considers a (static) complete graph with edge weights drawn independently and uniformly%
\footnote{The results also hold for any other distribution that closely matches the uniform distribution around~$0$.} from~$[0, 1]$.
Instead of arrival times, this model simply measures path lengths as sums of edge weights.
Intriguingly, some of our sharp thresholds have analogs with the same numerical values in that model.
Part of the underlying reason is that,
as shown in our analysis in \cref{lem:ExpOfXk},
if one reaches either endpoint of any particular edge at any point in time which is sufficiently close to~$0$,
then the waiting time until that edge appears
is distributed close to uniformly on~$[0, 1]$.
The latter condition matches the assumption made by \citeauthor{Janson99} about the distribution of the edge weights in his model and causes the growth of single-source trees to behave identically in both models.
This explains why the time thresholds for Point-to-Point Reachability and Temporal Source 
in our model agree with the weight thresholds for One-to-One Shortest Path and 
One-to-All Shortest Paths in Janson's model, respectively.

However, the above argument is not sufficient to account for the coincidence of the thresholds
involving multiple sources, \ie the time threshold for Temporal Connectivity and
the weight threshold for All-to-All Shortest Paths.
Indeed,
the strong independence assumption in the randomly edge-weighted complete graph model 
does not hold even in the gossiping model with repeated interactions:
In the former, the time required to cross an edge (\ie the edge weight) is independent of all other
factors, whereas in temporal graphs and gossiping two temporal paths emanating from different
sources, in general, need to wait different amounts of time to cross the same edge,
as those depend on the times at which the paths reach an endpoint of that edge.
The rather surprising observation that the above two thresholds are nevertheless the same
is ultimately a consequence of the fact that
the determining factor for pairwise connectivity in each of the two models
turns out to be
the waiting time until at least one edge appears for every vertex.
This time is identical in both models as all vertices start at time~$0$.

In essence, one might say that our results prove the temporal dependencies present in random temporal graphs to be insignificant for reachability times,
as witnessed by the fact that the resulting waiting time thresholds coincide with analogous thresholds in the simple weighted complete graph model, where such dependencies are absent.
Note, however, that our results aim in a different direction than those of \textcite{Janson99}.
Instead of characterizing the shape of the limit probability distribution,
we focus on the speed of convergence
and on the properties of partial foremost trees. Also, it is not clear that the aspects related to spanners have natural interpretations in Janson's model, as these graphs admit classical spanning trees unconditionally, which is not true in temporal graphs, whatever the density.

\subsection{More dependencies}

Our techniques can be applied to models in which edge appearances
are more dependent, and their distribution deviates even further from uniformity.
For instance, with some technical work we can show that the results obtained in this paper
also hold in a model where an adversary blocks a small fraction of the potential edges at each vertex
before any random choice is made.
Additionally, we could also allow an adversary to pick different time label
distributions for the edges as long as each of them is asymptotically close
to the uniform distribution around time~$0$.

\subsection{Edge-ordered graphs}
\label{sec:edge-ordered}

If only the relative order of time labels but not their absolute values is of interest,
simple temporal graphs with pairwise different time labels are clearly equivalent to
\emph{edge-ordered graphs}, which are graphs with a given total order of their edges.

In combinatorics, edge-ordered graphs have been studied from various perspectives.
Perhaps closest to our work is the study of the lengths of longest monotone 
paths and walks, i.e., paths or walks whose edge sequences are strictly increasing.
Note that in the terminology of temporal graphs, these are exactly the notions of temporal paths or walks.

Interest in this topic dates back to the '70s,
with \textcite{chvatal1971some} asking for worst-case bounds on the length of a longest monotone path or walk
in any edge ordering of the complete graph~$K_n$.
\Textcite{graham1973increasing} gave an answer for the case of walks 
and found a lower bound of~$\Omega(\sqrt{n})$ and an upper bound of~$3n/4$ for the case of paths.
The upper bound has subsequently been improved to roughly $n/2$ \cite{CalderbankCS84}.
The lower bound stayed current for a long time, but recently \textcite{bucic2020nearly} raised it to $n^{1-o(1)}$, nearly closing the gap.

More closely related to our paper is the work of \textcite{LL14}, who considered random edge orderings of~$K_n$
and found that even monotone Hamiltonian paths exist with probability at least $1/e$ and conjectured them to exist a.a.s.
This conjecture was recently proven true by \textcite{Martinsson19}.
Even more recently, \textcite{angel2020long} established the longest monotone walk in a random edge ordering of $G_{n,p}$ to \aas have length about $enp$.
These results can be interpreted also for temporal graphs.
Specifically the result of \citeauthor{Martinsson19} shows that any $\G \sim \F_{n,1}$ \aas contains a temporal path visiting all vertices,
while that of \citeauthor{angel2020long} implies that the length of a longest temporal walk in $\F_{n,p}$ is concentrated around~$enp$.
%

\section{Concluding remarks and open questions}
\label{sec:conclusion}

In this paper, we studied a natural model of random temporal graphs, in which every edge of an Erd{\H{o}}s–R{\'e}nyi graph $G_{n,p}$ is assigned a single presence time chosen uniformly at random in the unit interval $[0,1]$.
The study of various degrees of temporal reachability and existence of small temporal spanners 
in these graphs revealed a rich diversity of thresholds, in stark contrast with static graphs.
Put together, these thresholds offer a measure of the discrepancy between static and temporal connectivity in random graphs.

Despite the simplicity of RSTGs, we have shown that they capture, in a scale-preserving way, several classical phenomena observed in gossip theory and population protocols.
Furthermore, some of these results were shown to strengthen and/or to complete existing results in these fields.
In stark contrast with \emph{deterministic} temporal graphs, where significant obstructions exist which prevent the existence of linear-size spanners even in some dense graphs, we have shown that all possible obstructions to the existence of nearly optimal temporal spanners turn out to be statistically insignificant in random temporal graphs.
Thus, the concept of a nearly optimal spanner recovers some form of universality in the spirit of spanning trees in static graphs.
Observe that the spanners we construct are highly centralized, with almost all communication passing through a single pivot vertex.
It would be interesting to determine whether this is inherent or incidental.

All but one of our characterizations are sharp thresholds.
The remaining one concerns the existence of optimal temporal spanners (\ie spanners of size exactly~$2n-4$) in RSTGs, for which we prove an upper bound of $p=4 \log n / n$.
Whether this bound is actually a sharp threshold is left open; we conjecture that it is (Conjecture~\ref{conjecture}).

Now that the most basic phase transitions in temporal reachability are characterized in RSTGs, it would be quite natural to start looking at more complex properties, e.g., motivated by networking applications.
On the one hand, one might want to obtain connectivity via at least~$2$ 
(or~$k$) disjoint temporal paths, where disjoint can mean edge-disjoint
or vertex-disjoint.
On the other hand, one could extend the task of sending a message to the task
of sending a message and receiving a reply, \ie roundtrip communication.
We conjecture that there are similar sharp thresholds for roundtrip communication. 
In particular, our preliminary investigation suggests the following:

\begin{conjecture}[Roundtrip connectivity]~
 \begin{enumerate}
         \item For two fixed vertices, the sharp threshold for possibility
                 of roundtrip communication is $2\log n / n$;
         \item For one fixed vertex $u$, the sharp threshold for roundtrip communication from~$u$ 
                 to all other vertices and back, as well as the sharp threshold for roundtrip communication
                 from all the other vertices to $u$ and back, is $3\log n / n$;
\item  The sharp threshold for roundtrip communication between all pairs of vertices (\ie roundtrip connectivity) is $4\log n / n$.
 \end{enumerate}
\end{conjecture}

A classical concept in Erd{\H{o}}s–R{\'e}nyi random graphs is the one of
a \emph{giant component}, \ie a set of $\Theta(n)$ mutually reachable vertices. In this paper, we showed that $\log n / n$ is a sharp threshold for any fixed pair of vertices to reach each other via temporal paths. However, in a subtle way, this property does not imply that a subset of $\Theta(n)$ vertices can \emph{all} reach each other (due to the intransitivity of temporal paths). Nevertheless, we conjecture that such components also appear at $\log n / n$ and that the threshold is sharp, in contrast to static graphs where the giant component emerges progressively.

\begin{conjecture}
  $\log n / n$ is a sharp threshold for the existence of a temporal component of size~$\Theta(n)$. In fact, the same holds even for components of size $n-o(n)$.
\end{conjecture}

Finally, we have shown in the paper that in an RSTG, every pair of temporally connected vertices is connected by a temporal path containing at most a logarithmic number of edges.
We conjecture that for sufficiently sparse RSTGs this also constitutes a lower bound. Precisely:

\begin{conjecture}
There exists a positive-valued function $f$
such that 
for any constant $C$, if $p=C \log n/n$, then the graph \aas contains some pairs of vertices
with no connecting temporal path having fewer than $f(C)\log n$ edges.
\end{conjecture}

In a dedicated section, we have shown that our analyses can be applied essentially in the same way for multiple types
of processes beyond RSTGs.
We hope that this versatility can be extended further and characterized
by some simple conditions.
This was the case with \emph{parallel} rumor spreading time \cite{DoerrK17},
and we hope that sequential rumor propagation could enjoy the same.

In conclusion, the notion of reachability in temporal graphs is quite different from reachability in static graphs, and the present paper illustrates the fact that reachability in \emph{random} temporal graphs is also quite different from reachability in \emph{deterministic} temporal graphs. We hope that our results will help paving the way for further (and more complex) investigations of temporal reachability.

\smallskip
\bigskip
\noindent
\textbf{Acknowledgements.}
We are very grateful to the anonymous reviewers for their careful reading and many useful and detailed comments that improved the presentation of the paper.

\printbibliography

\end{document}